\documentclass[10pt]{article}
\usepackage{amssymb,amsmath,xcolor}
\usepackage{amssymb,amsmath}
\usepackage{amsmath}

\usepackage{yfonts}

\usepackage{enumitem}
\usepackage{adjustbox}
\usepackage{thmtools,mathtools}
\usepackage{ulem}

\usepackage[utf8]{inputenc}
\usepackage{stackrel}
\usepackage{dcpic}
\usepackage{pictexwd}
\usepackage{fixltx2e}
\usepackage{enumerate}
\usepackage{relsize}
\usepackage{stmaryrd}
\usepackage{cancel}
\usepackage{yfonts}
\usepackage{xcolor}
\usepackage{csquotes }
 \usepackage{amsthm}
  \usepackage{enumitem} 
\DeclareRobustCommand\mos[1]{\mathrel{|}\joinrel
\stackrel{#1}{\mathrel{=}}}
\newcommand{\mosn}[1]{\not \mos{#1}}

\DeclareRobustCommand\vds[1]{\,\mathrel{|}\joinrel\joinrel\joinrel
\frac{#1}{ \ \ \ }}

\DeclareRobustCommand\vss[1]{\,\mathrel{|}\joinrel
\stackrel{#1}{\sim}}

\newcommand{\mb}{\scriptscriptstyle{\Box}}
\newcommand{\lra}{\leftrightarrow}
\usepackage{bm,listings,todonotes}

\DeclareBoldMathCommand\boldlangle{\left\langle}
\DeclareBoldMathCommand\boldrangle{\right\rangle}

\newcommand{\cf}[1]{\boldlangle #1 \boldrangle}
\newcommand{\Sset}{\textit{set}}
\newcommand{\Smem}{\textit{mem}}

\newcommand{\keyword}[1]{\mathrel{\textbf{#1}}}

\newcommand{\Swhile}{\keyword{while}}
\newcommand{\Sdo}{\keyword{do}}
\newcommand{\Sif}{\keyword{if}}
\newcommand{\Sthen}{\keyword{then}}
\newcommand{\Selse}{\keyword{else}}
\newcommand{\Sskip}{\keyword{skip}}
\newcommand{\Sleq}{\keyword{leq}}
\newcommand{\Sasgn}{\keyword{asgn}}
\newcommand{\Splus}{\keyword{plus}}
\newcommand{\Strue}{\keyword{true}}
\newcommand{\Sfalse}{\keyword{false}}
\newcommand{\nom}{N}
\newcommand{\ora}{\mathbf}

\newcommand{\Spgm}{{\it pgm}}
\newcommand{\SB}{{\it Cnd}}
\newcommand{\SC}{{\it Body}}

\theoremstyle{plain}
\newtheorem{theorem}{Theorem}[section]
\newtheorem{lemma}[theorem]{Lemma}

\newtheorem{remark}[theorem]{Remark}
\newtheorem{proposition}[theorem]{Proposition}
\newtheorem{definition}[theorem]{Definition}

\newtheorem{claim}{Claim}

\begin{document}
\normalem

\title{Operational semantics and program verification using many-sorted hybrid modal logic}

\author{Ioana Leu\c stean \and
Natalia Moang\u a \and
Traian Florin \c Serb\u anu\c t\u a\\ 
{\small Faculty of Mathematics and Computer Science, University of Bucharest,}\\
{\small Academiei nr.14, sector 1, C.P. 010014,  Bucharest, Romania} \\
{\small ioana@fmi.unibuc.ro}  \hspace*{0.3cm}  {\small natalia.moanga@drd.unibuc.ro}\\
{\small traian.serbanuta@fmi.unibuc.ro}
}  
\date{}
  
\maketitle              
\begin{abstract}

We propose a general framework to allow: (a) specifying the operational semantics of a programming language; and (b) stating and proving properties about program correctness.  
Our framework is based on a many-sorted system of hybrid modal logic, for which we prove completeness results.
We believe that our approach to program verification improves over the existing approaches within modal logic as (1) it is based on operational semantics which allows for a more natural description of the execution than Hoare's style weakest precondition used by dynamic logic; (2) being multi-sorted, it allows for a clearer encoding of semantics, with a smaller representational distance to its intended meaning.
\medskip

\end{abstract}

\section{Introduction}
\label{sec:intro}

Program verification within \textit{modal logic}, as showcased by \textit{dynamic logic} \cite{dynamic}, is following the mainstream axiomatic approach proposed by Hoare/Floyd \cite{floyd,hoare}. In this paper, we continue our work from \cite{noi} in exploring the amenability of dynamic logic in particular, and of modal logic in general, to express operational semantics of languages (as axioms), and to make use of such semantics in program verification. 
Consequently, we consider the SMC Machine described by Plotkin \cite{plotkin}, we derive a dynamic logic set of axioms from its proposed transition semantics, and we argue that this set of axioms can be used to derive Hoare-like assertions regarding functional correctness of programs written in the SMC language.

The main idea is to define a general logical system that is powerful enough to represent both the programs and their semantics in a uniform way. With respect to this, we follow the line of \cite{goguen} and the recent work from \cite{rosu}.
 
The logical system that we developed as support for our approach is a {\em many-sorted hybrid polyadic modal logic}, built upon our general many-sorted polyadic modal logic defined in \cite{noi}. We chose a modal setting since, as argued above, through dynamic logic and Hoare logic, modal logic has a long-standing tradition in program verification (see also \cite{popl} for a  modal logic approach to separation logic \cite{separation}) and it is successfully used in specifying and verifying hybrid systems \cite{platzerbook}.   
   
In \cite{noi} we defined a general many-sorted modal logic, generalizing some of the already existing approaches, e.g. \cite{manys6,manys2} (see \cite{noi} for more references on many-sorted modal logic). This system allows us to specify a language and its operational semantics and one can use it to  certify executions as well. However, both its expressivity and its capability  are limited: we were not able to perform symbolic execution and, in particular, we were not able to prove Hoare-style invariant properties for loops. In Remark \ref{rem1}, we point out some theoretical aspects related to these issues.   

In the present paper we employ the procedure of {\em hybridization} on top of our many-sorted modal logic previously defined. We drew our inspiration from 
\cite{platzerhybrid,rosu} for practical aspects, and from the extensive research on {\em hybrid modal logic} \cite{hand,mod} on the theoretical side. 

Our aim was to develop a system that is strong enough to perform all the addressed issues (specification, semantics, verification), but also to keep it as simple as possible from a theoretical point of view. To conclude:  in our setting we are able to associate a sound and complete many-sorted hybrid modal logic to a given language such that both operational semantics and program verification can be performed through logical inference.

We have to make a methodological comment: sometimes nominals are presented as another {\em sort} of atoms (see, e.g.\cite{mod}). Our sorts come from a many-sorted signature $(S,\Sigma)$, as in \cite{goguen}, so  all the formulas (in particular the propositional variables, the state variables, the nominals) are $S$-sorted sets. When we say that the hybrid logic is {\em mono-sorted} we use sorted according to our context, i.e. the sets of propositional variables, nominals and state variables are regular sets and not $S$-sets. 

We recall our many-sorted modal logic \cite{noi} in Section \ref{sec:prel}.  The hybridization is performed in Section \ref{sec:hybridgen}. A concrete language and its operational semantics are defined in Section \ref{sec:app}; we also show how to perform Hoare-style verification. A section on related and future work concludes our paper.

\section{Preliminaries: a many-sorted  modal logic}
\label{sec:prel}


Our language is determined by a fixed, but arbitrary,  many-sorted signature ${\bf \Sigma}=(S, \Sigma)$ and an $S$-sorted set of  propositional variables $P=\{P_s\}_{s\in S}$ such that  $P_s\neq \emptyset$ for any $s\in S$ and  $P_{s_1} \cap P_{s_2} = \emptyset $ for any  $s_1 \neq s_2$ in $S$.
For any $n\in {\mathbb N}$ and $s,s_1,\ldots, s_n\in S$ we denote 
$\Sigma_{s_1\ldots s_n,s}=\{\sigma\in \Sigma\mid \sigma:s_1 \cdots s_n\to s\}$.

The set of formulas of $\mathcal{ML}_{\bf \Sigma}$ is an $S$-indexed family  inductively defined by:
 \vspace*{-0.2cm}
 \begin{center}
 $\phi_s ::=  p\,|\,\neg\phi_s\,|\,{\phi}_s\vee {\phi}_s\,|\, \sigma ({\phi}_{s_1}, \ldots , {\phi}_{s_n} )$
\end{center}
\vspace*{-0.2cm}
where $s\in S$, $p\in P_s$ and  $\sigma \in \Sigma_{s_1 \cdots s_n,s}$.

We use the classical definitions of the derived logical connectors: for any $\sigma\in \Sigma_{s_1 \ldots s_n,s}$ the {\em dual operation} is
$\sigma^{\mb} (\phi_1, \ldots , \phi_n ) := \neg\sigma (\neg\phi_1, \ldots , \neg\phi_n ).$

In the sequel,  by $\phi_s$ we mean that $\phi$ is a formula of sort $s\in S$. Similarly, $\Gamma_s$ means that $\Gamma$ is a set of formulas of sort $s$. When the context uniquely determines  the sort of a state symbol, we shall omit the subscript. 

In order to define the semantics we introduce the $(S,\Sigma)$\textit{-frames} and the $(S,\Sigma)$\textit{-models}.
 An $(S,\Sigma)${\em -frame} is a tuple 
 $\mathcal{F} =({W},(R_\sigma)_{\sigma\in\Sigma})$
 such that:
\begin{itemize}
\item  ${W} =\{ W_s \}_{s\in S}$ is an  $S$-sorted set of worlds and $W_s\neq \emptyset$ for any $s\in S$,
\item  ${R}_{\sigma} \subseteq  W_s \times W_{s_1} \times \ldots \times W_{s_n}$  for any $\sigma \in \Sigma_{s_1 \cdots s_n,s}$.
\end{itemize}

An $(S,\Sigma)$-{\em model based on} $\mathcal{F}$ is a pair  ${\mathcal M}= ({\mathcal F},V)$ where $V =\{V_s\}_{s\in S}$ such that $V_s : P_s \to \mathcal{P}(W_s)$ for any $s\in S$. The model $\mathcal{M}= (\mathcal{F},V)$ will be simply denoted as  $\mathcal{M}= ({W}, (R_\sigma)_{\sigma\in\Sigma},V)$. 

In the sequel we introduce a many-sorted  \textit{satisfaction relation}.
If $\mathcal{M}= (W,(R_\sigma)_{\sigma\in\Sigma},V)$ is an $(S,\Sigma)$-model, $s\in S$, $w\in W_s$ and  $\phi$ is a formula of sort $s$, then the  many-sorted  \textit{satisfaction relation} $\mathcal{M},w\mos{s} \phi$
is inductively defined:
\begin{itemize}
\item $\mathcal{M},w \mos{s} p$ iff $w\in V_s(p)$
\item $\mathcal{M},w \mos{s} \neg \psi$ iff $\mathcal{M},w \not\mos{s}\psi$
\item $\mathcal{M},w \mos{s} \psi_1 \vee \psi_2$ iff $\mathcal{M},w \mos{s} \psi_1$ or $\mathcal{M},w \mos{s} \psi_2$ 
\item if $\sigma\in \Sigma_{s_1 \ldots s_n,s} $,  then $\mathcal{M},w \mos{s} \sigma(\phi_1, \ldots , \phi_n )$ iff there exists  $(w_1,\ldots,w_n) \in W_{s_1}\times\cdots\times W_{s_n}$  such that $\mathcal{R}_{\sigma} ww_1\ldots w_n$ and
 $\mathcal{M},w_i  \mos{s_i} \phi_i$ for any $i \in  [n] $.
\end{itemize}

\begin{definition}[Validity and satisfiability] Let $s\in S$ and assume $\phi$ is a formula of sort $s$. Then $\phi$ is {\em satisfiable} if  ${\mathcal M},w\mos{s}\phi$ for some model $\mathcal M$ and some $w\in W_s$.  The formula $\phi$ is {\em valid} in a model $\mathcal M$ if ${\mathcal M},w\mos{s}\phi$ for any $w\in W_s$; in this case we write ${\mathcal M}\mos{s}\phi$. The formula $\phi$ is valid in a  frame $\mathcal F$  if $\phi$ is valid in all the models based on $\mathcal F$; in this case we write ${\mathcal F}\mos{s}\phi$. Finally, the formula $\phi$ is valid if $\phi$ is valid in all frames; in this case we write $\mos{s}\phi$. 
\end{definition}

\begin{figure}[h]
\centering
{\small
{\bf The system   ${\mathbf K}_{\bf\Sigma}$}

\begin{itemize}
\item For any $s\in S$, if $\alpha$ is a formula of sort $s$ which is a theorem in propositional logic, then $\alpha$ is an axiom. 
\item Axiom schemes: for any $\sigma\in \Sigma_{s_1\cdots s_n,s}$ and for any formulas $\phi_1,\ldots, \phi_n,\phi,\chi$ of appropriate sorts, the following formulas are axioms:

\begin{tabular}{rl}
$(K_\sigma)$ & $\sigma^{\mb}(\ldots,\phi_{i-1},\phi\rightarrow\chi,\phi_{i+1}, \ldots)\to$\\ &\hspace*{0.5cm}$( \sigma^{\mb}(\ldots ,\phi_{i-1}, \phi, \phi_{i+1},\ldots) \to \sigma^{\mb}(\ldots ,\phi_{i-1}, \chi, \phi_{i+1},\ldots))$\\
$(Dual_\sigma)$& $\sigma (\psi_1,\ldots ,\psi_n )\leftrightarrow \neg \sigma^{\mb} (\neg \psi_1,\ldots ,\neg \psi_n )$
\end{tabular}

\item Deduction rules: {\em Modus Ponens} and {\em Universal Generalization}
\medskip

\begin{tabular}{rl}
$(MP)$ & if $\vds{s}\phi$ and $\vds{s}\phi\to \psi$ then 
$\vds{s}\psi$\\
$(UG)$ &  if $\vds{s_i}{\phi}$ then $\vds{s}\sigma^{\mb} (\phi_1, .. ,\phi, ..\phi_n)$
\end{tabular}
\end{itemize} }
\caption{$(S,\Sigma)$ modal logic}\label{fig:k}
\end{figure}

The {\em set of theorems} of ${\mathbf K}_{\bf \Sigma}$ is the least set of formulas that contains all the axioms and it is closed under deduction rules. Note that the set of theorems is obviously closed under {\em $S$-sorted uniform substitution} (i.e. propositional variables of sort $s$ are uniformly replaced by formulas of the same sort). If $\phi$ is a theorem of sort $s$ write \mbox{$\vds{s}_{{\mathbf K}_{\bf\Sigma}}\phi$.} Obviously, ${\mathbf K}_{\bf \Sigma}$ is a  generalization of the modal system $\mathbf K$ (see \cite{mod} for the mono-sorted version).  

The distinction between local and global deduction from the mono-sorted setting (see \cite{mod}) is deepened in our version: \textit{locally}, the conclusion and the hypotheses have the same sort, while \textit{globally}, the set of hypotheses is a many-sorted set. In the sequel we only consider the local setting. 

\begin{definition}[Local deduction]\cite{noi} If $s\in S$ and $\Gamma_s\ \cup\ \{\phi\}$  is a set of formulas of sort $s$, then we say that $\phi$ is ({\em locally}) {\em provable from} $\Gamma_s$  if there are $\gamma_1, \ldots, \gamma_n\in \Gamma_s$ such that   
\mbox{$\vds{s}_{{\mathbf K}_{\bf\Sigma}}(\gamma_1\wedge \ldots \wedge \gamma_n) \to \phi$.} In this case we write \mbox{$\Gamma_s \vds{s}_{{\mathbf K}_{\bf\Sigma}}\phi$.} 
\end{definition}

The  construction of the canonical model is a straightforward generalization of the mono-sorted setting. For more details, we refer to \cite{noi}. The last result we recall is the (strong) completeness theorem with respect to the class of all frames. 

\begin{theorem}\cite{noi} Let $\Gamma_s$ be a set of formulas of set $s$. If $\Gamma_s$ is a consistent set in ${\mathbf K}_{\bf \Sigma}$ then 
$\Gamma_s$ has a model. Moreover, if $\phi$ is a formula of sort $s$, then $\Gamma_s\models_{{\mathbf K}_{\bf\Sigma}}\phi$ {\mbox iff } $\Gamma_s\vdash_{{\mathbf K}_{\bf\Sigma}}\phi$,
where $\Gamma_s\models_{{\mathbf K}_{\bf\Sigma}}\phi$ denotes the fact that any model of $\Gamma$ is also a model of $\phi$. 
\end{theorem}

\begin{remark}[Problems]\label{rem1}
The many-sorted modal logic allows us to define both the syntax and the semantics of a programming language (see \cite{noi} for a complex example). However, there are few issues, both theoretical and operational, that we could not overcome:
\vspace*{-0.2cm} 
 \begin{itemize}
 \item[(i1)] the logic can be used to certify executions, but not to perform symbolic verification; in particular, in order to prove the invariant properties for loops, the existential binder is required; 
\item[(i2)] the completeness theorem for extensions of 
${\mathbf K}_{\bf\Sigma}$ from \cite{noi} only refers to model completeness, but says nothing about frame completeness (see \cite{goranko} for a general discussion  on this distinction); 
\item[(i3)] the sorts are completely isolated formally, but in our example elements of different sorts have a rich interaction.
 \end{itemize}
These issues will be adressed in the following sections.
 \end{remark}

\section{Many-sorted hybrid modal logic}
\label{sec:hybridgen}

The hybridization of our many-sorted modal logic is developed  using a combination of ideas and techniques from \cite{hand,pureax,hyb,mod,goranko,goranko2}. Hybrid logic is defined on top of modal logic by  adding {\em nominals}, {\em states variables} and specific operators and binders. 

Nominals  allow us to directly refer the worlds (states) of a model, since they are evaluated in singletons in any model. However, a nominal may refer different worlds in different models. In the sequel we  introduce the {\em constant nominals}, which are evaluated to singletons, but  they refer to the same world (state) in all models.  Our example for constant nominals are \texttt{true} and \texttt{false} from  Section \ref{sec:app}. 

\begin{definition}[Signature with constant nominals]
A {\sf signature with constant nominals} is a triple $(S,\Sigma, {\nom})$ where $(S,\Sigma)$ is a many-sorted signature and ${\nom}=({\nom}_s)_{s\in S}$ is an $S$-sorted set of constant nominal symbols. In the sequel, we denote 
${\bf \Sigma}=(S,\Sigma, {\nom})$. 
\end{definition}

As before, the sorts will be denoted by $s$, $t$, $\ldots$ and by ${\rm PROP}=\{{\rm PROP}_s\}_{s\in S}$, ${\rm NOM}=\{{\rm NOM}_s\}_{s\in S}$ and ${\rm SVAR}=\{{\rm SVAR}_s\}_{s\in S}$ we will denote some countable $S$-sorted sets. The elements of ${\rm PROP}$ are ordinary propositional variables and they will be denoted $p$, $q$,$\ldots$; the elements of ${\rm NOM}$ are called {\em nominals} and they will be denoted by $j$, $k$, $\ldots$; the elements of ${\rm SVAR}$ are called {\em state variables} and they are denoted $x$, $y$, $\ldots$.  We shall assume  that for any distinct sorts $s\neq t\in S$, the corresponding sets of propositional variables, nominals and state variables are distinct.    A {\em state symbol} is a nominal, a constant nominal or a state variable.. 

As in the mono-sorted case, nominals and state variables will be semantically constrained: they are  evaluated in singleton, which means they will  always refer to a unique world of our model. In addition, the constant nominals will refer to the same world(state) in any evaluation, so they will be defined at the frames' level.

In the mono-sorted setting, starting with a modal logic,  the  simplest hybrid system  is obtained by adding nominals alone. However, the {\em basic hybrid system} is obtained by adding the {\em satisfaction modality} $@_j\phi$ (which states that $\phi$ is true at the world denoted by the nominal $j$). The most powerful hybrid systems are obtained by further adding the binders $\forall$ and $\exists$ that bind state variables to worlds, with the expected semantics \cite{hand,blsel,hyb}. The subsequently defined systems ${\mathcal H}_{\bf \Sigma}(@)$ and ${\mathcal H}_{\bf \Sigma}(@,\forall)$ develop the hybrid modal logic in our many-sorted setting. 

\medskip

Note that, whenever the context is clear, we'll simply write  $\mos{s}$ instead of $\mos{s}_{{\mathcal H}_{\bf \Sigma}(@)}$  or $\mos{s}_{{\mathcal H}_{\bf \Sigma}(@,\forall)}$, and $\vds{s}$  instead of $\vds{s}_{{\mathcal H}_{\bf \Sigma}(@)}$  or $\vds{s}_{{\mathcal H}_{\bf \Sigma}(@,\forall)}$. We will further assume that the sort of a formula (set of formulas) is implied by a concrete context but, whenever  necessary, we will use subscripts to fix the sort of a symbol: $x_s$ means that $x$ is a state variable of sort $s$, $\Gamma_s$ means that $\Gamma$ is a set of formulas of sort $s$, etc. 

\newpage
\begin{definition}[Formulas] For any $s\in S$ we define the formulas of sort $s$:\\
\begin{tabular}{ll}
- for ${\mathcal H}_{\bf \Sigma}(@)$: & 
$\phi_s :=  p\mid j\mid \neg \phi_s \mid  \phi_s \vee \phi_s \mid  \sigma(\phi_{s_1}, \ldots, \phi_{s_n})_s \mid  @_k^s\phi_t$\\
- for ${\mathcal H}_{\bf \Sigma}(@,\forall)$: & 
$\phi_s :=  p\mid j\mid y_s\mid \neg \phi_s \mid\phi_s \vee \phi_s \mid \sigma(\phi_{s_1}, \ldots, \phi_{s_n})_s
 \mid  @_k^s\phi_t \mid  \forall x_t\, \phi_s$
\end{tabular} 
\medskip

\noindent Here, $p\in {\rm PROP}_s$, $j\in {\rm NOM}_s\cup {\nom}_s$, $t\in S$, $k\in {\rm NOM}_t\cup{\nom}_t$,$x\in {\rm SVAR}_t$, $y\in {\rm SVAR}_s$ and  $\sigma\in \Sigma_{s_1\cdots s_n,s}$. For any $\sigma\in \Sigma_{s_1\ldots s_,s}$, the dual formula $\sigma^{\mb}(\phi_1,\ldots, \phi_n)$ is defined as in Section \ref{sec:prel}. We also define the  {dual binder} $\exists$: for any $s,t\in S$, if $\phi$ is a formula of sort  $s$ and $x$ is a state variable of sort $t$, then  $\exists x\, \phi := \neg\forall x\, \neg\phi$ is a formula of sort $s$.  The notions of {\sf free state variables} and {\sf bound state variables} are defined as usual.  
\end{definition}

\begin{remark}[Expressivity]\label{rem3}
As a departure from our sources of inspiration, we only defined the satisfaction operators $@_j$ for nominals, and not for state variables. Hence, $@_x$ is not a valid formula in our logic. Our reason was to keep the system as "simple" as possible, but strong enough to overcome the problems encountered in the non-hybrid setting (see Remarks \ref{rem1}). More issues concerning expressivity are analyzed in  Section \ref{sec:final}.   
\end{remark}

One important remark is the definition of the satisfaction modalities: 
{\em if $k$ and $\phi$ are a nominal and respectively, a formula of the sort $t\in S$, then we define a family of satisfaction operators 
$\{@_k^s\phi\}_{s\in S}$} such that $@_k^s\phi$ is a formula of sort $s$ for any $s\in S$. This means that $\phi$ is true at the world denoted by $k$ on the sort $t$ and is acknowledged on any sort $s\in S$. So, our sorted worlds are not isolated anymore, both from a syntactic and a semantic point of view.

\begin{definition}
If  ${\bf \Sigma}=(S,\Sigma, {\nom})$ then a ${\bf \Sigma}$-{\sf frame} is ${\mathcal F}=(W, (R_\sigma)_{\sigma \in \Sigma}, {\nom}^{\cal F})$  where $(W, (R_\sigma)_{\sigma \in \Sigma})$ is an $(S,\Sigma)$-frame and ${\nom}^{\cal F}=({\nom}_s^{\cal F})_{s\in S}$ and ${\nom}_{s}^{\cal F}=(w^c)_{c\in {\nom}_s}$ $\subseteq W_s$ for any $s\in S$. We will further assume that distinct constant nominals have distinct sorts, so we shall simply write ${\nom}^{\cal F}=(w^c)_{c\in {\nom}}$.
\end{definition}

\begin{definition}[The satisfaction relation in  ${\mathcal H}_{\bf \Sigma}(@)$]
A {\em (hybrid) model} in ${\mathcal H}_{\bf \Sigma}(@)$ is a triple
\vspace*{-0.2cm}
\begin{center}
${\mathcal M}=(W, (R_\sigma)_{\sigma\in\Sigma}, (w^c)_{c\in {\nom}}, V)$ 
\end{center}

\noindent where 
$V:{\rm PROP}\cup{\rm NOM}\to {\mathcal P}(W)$ is an $S$-sorted valuation such that $V_s(k)$ is a singleton for any $s\in S$ and $k\in {\rm NOM}_s$.  If $V$ is an $S$-sorted evaluation, we define $V^{\nom}:{\rm PROP}\cup {\rm NOM}\cup \nom\to {\cal P}(W)$ by $V^{\nom}_s(c)=\{w^c\}$ for any $s\in S, c\in{\nom}_s$ and $V^{\nom}_s(v)=V_s(v)$ otherwise. 

 The satisfaction relation  for nominals, constant nominals and satisfaction operators is defined as follows:
\begin{itemize}
\item $\mathcal{M},w \mos{t} k$ if and only if $V^{\nom}_t(k)=\{w\}$,
\item $\mathcal{M},w' \mos{s} @_k^s\phi$ if and only if $\mathcal{M},w \mos{t} \phi$ where $V^{\nom}_t(k)=\{w\}$.
\end{itemize}
Here $s,t\in S$,  $w\in W_t$, $w'\in W_s$,  $k\in {\rm NOM}_t\cup {\rm {\nom}}_t$ and $\phi$ is a formula of sort $t$. 
\end{definition}
\noindent Satisfiability and validity in ${\mathcal H}(@)$ are defined  as in  Section \ref{sec:prel}.

\medskip

In order to define the semantics for ${\mathcal H}_{\bf \Sigma}(@,\forall)$ more is needed. Given a model ${\mathcal M}=(W, (R_\sigma)_{\sigma\in\Sigma}, (w^c)_{c\in\nom}, V)$, an {\em assignment} is an $S$-sorted  map $g : {\rm SVAR} \rightarrow W$.  If $g$ and $g'$ are assignment functions $s\in S$ and $x\in SVAR_s$ then we say that $g'$ is an {\em $x$-variant} of $g$ (and we write $g'\stackrel{x}{\sim} g$)  if $g_t=g'_t$ for 
$t\neq s\in S$ and  $g_s(y)=g'_s(y)$ for any $y\in SVAR_s$, $y \neq x$. 

\begin{definition}[The satisfaction relation in  ${\mathcal H}_{\bf \Sigma}(@,\forall)$] In the sequel
\vspace*{-0.2cm}
\begin{center}
 ${\mathcal M}=(W, (R_\sigma)_{\sigma\in\Sigma}, (w^c)_{c\in\nom}, V)$
\end{center}
\vspace*{-0.2cm} 
 is a model and $g:{\rm SVAR}\to W$ an $S$-sorted assignment. The satisfaction relation is defined as follows:
\begin{itemize}
\item $\mathcal{M},g,w \mos{s} a$, if and only if $w\in V^{\nom}_s(a)$, where $a\in {\rm PROP_s}\cup {\rm NOM_s}\cup {\nom}_s$,
\item $\mathcal{M},g,w \mos{s} x$, if and only if $w=g_s(x)$, where $x\in {\rm SVAR}_s$,
\item $\mathcal{M},g,w \mos{s} \neg \phi$, if and only if $\mathcal{M},g,w \mosn{s}\phi$
\item $\mathcal{M},g,w \mos{s} \phi \vee \psi$, if and only if $\mathcal{M},g,w \mos{s} \phi$ or $\mathcal{M},g,w \mos{s} \psi$ 
\item if $\sigma\in\Sigma_{s_1\ldots s_n,S}$ then $\mathcal{M},g,w \mos{s} \sigma(\phi_1, \ldots , \phi_n )$, if and only if there is \\$(w_1,\ldots,w_n) \in W_{s_1}\times\cdots\times W_{s_n}$ such that  $R_{\sigma} ww_1\ldots w_n$ and $\mathcal{M},g,w_i  \mos{s_i} \phi_i$ for any $i \in [n]$,
\item $\mathcal{M},g,w \mos{s} @_k^s\phi$ if and only if $\mathcal{M},g,u\mos{t} \phi$ where $k\in {\rm NOM}_t\cup {\nom}_t$,  $\phi$ has the sort $t$ and $V^{\nom}_t(k)=\{u\}$,
 \item $\mathcal{M},g,w \mos{s} \forall x\,\phi$, if and only if  $\mathcal{M},g',w \mos{s} \phi$ for all $g'\stackrel{x}{\sim} g$.

Consequently, 
\item $\mathcal{M},g,w \mos{s} \exists x\, \phi$, if and only if $\exists g'( g' \stackrel{x}{\sim} g \ and \  \mathcal{M},g',w \mos{s} \phi)$.
\end{itemize}
 \end{definition}

Following the mono-sorted setting, satisfiability in ${\mathcal H}(@,\forall)$ is defined as follows:  a formula $\phi$ of sort $s\in S$ is {\em satisfiable} if  ${\mathcal M},g, w\mos{s}\phi$ for some model $\mathcal M$, some assignment $g$  and some $w\in W_s$. Consequently,   the formula $\phi$ is {\em valid} in a model $\mathcal M$ if ${\mathcal M},g, w\mos{s}\phi$ for any assignment $g$ and any  $w\in W_s$. One can speak about validity in a  frame as in  Section \ref{sec:prel}.

In the presence of nominals, we can speak about  {\em named models} and {\em pure formulas}, as in \cite{mod}[Section 7.3].

\begin{definition}[Named models and pure formulas]  A formula is {\sf pure} if it does not contain propositional variables.  A {\sf pure instance}  of a formula is obtained by  is obtained by uniformly substituting nominals for nominals of the same sort. A model ${\mathcal M}=(W, (R_\sigma)_{\sigma\in\Sigma},(w^c)_{c\in{\nom}} V)$ is {\sf named} if for any sort $s\in S$ and  world $w\in W_s$ there exists $k\in {\rm NOM}_s\cup{\rm {\nom}}_s$ such that $w=V^{\nom}_s(k)$.
\end{definition}

As in the mono-sorted case,  pure formulas and named models are important since they give rise to strong completeness results with respect to the class of frames they define.  

\begin{proposition}[Pure formulas in ${\mathcal H}_{\bf \Sigma}(@)$]\label{prop:pure1}
Let ${\mathcal M}=(W, (R_\sigma)_{\sigma\in\Sigma}, (w^c)_{c\in {\nom}}, V)$  be a named model, ${\mathcal F}=(W, (R_\sigma)_{\sigma\in\Sigma}, (w^c)_{c\in{\nom}})$ the corresponding frame and $\phi$ a pure formula of sort $s$. Then
${\mathcal F}\mos{s}\phi$ if and only if ${\mathcal M}\mos{s}\psi$ for  any $\psi$ that is a pure instance of $\phi$.    
\end{proposition}
\begin{proof}
Let $\phi$ be a pure formula of sort $s$ and suppose $\mathcal{F}  \mosn{s} \phi$. Then there exist a valuation $V'$  and some state $w \in W_s$ in the model $\mathcal{M}'=(\mathcal{F}, V')$ such that $\mathcal{M}',w  \mosn{s} \phi$.

On each sort $s \in S$ we will notate $j^s_1, \ldots, j^s_t$ all the nominals occurring in $\phi$. But because we are working in a named model, $V$ labels every state of any sort in $\mathcal{F}$ with a nominal of the same sort. Hence, on each sort $s \in S$ there exist $k^s_1, \ldots, k^s_t$ nominals such that $V_s^{\nom}(j^s_1)=V'_s(k^s_1)$, $\ldots$ ,$V_s^{\nom}(j^s_t)=V'_s(k^s_t)$. Therefore, if $\mathcal{M}', w \mosn{s} \phi$ and $\psi$ is obtained by substituting on each sort each nominal $j^s_i$ with the corresponding one $k^s_i$, then $\mathcal{M}, w \mosn{s} \psi$.

But $\phi$ is a pure formula, and by substituting the nominals contained in the formula with other nominals of the same sort, the new instance it is also a pure formulas like $\psi$. Therefore, by hypothesis, we have $\mathcal{M},v \mos{s} \psi$ for any $v \in W_s$. But also $w\in W_s$, hence $\mathcal{M},w \mos{s} \psi$, and we have a contradiction.
\end{proof}

Can we prove a similar result for the system ${\mathcal H}_{\bf \Sigma}(@,\forall)$? We give a positive answer to this question, inspired by the discussion on existential saturation rules from \cite{pureax}[Lemma 1]. In order to do this, we define the {\em $\forall$-pure} formulas and we characterize frame satisfiability for such formulas. 
As consequence,  Propositions \ref{prop:pure1} and \ref{prop:pure2} will lead to completeness results with respect to frame validity.

\begin{definition}
In ${\mathcal H}_{\bf \Sigma}(@,\forall)$, we say that a formula is {\em $\forall$-pure} if it is pure or it has the form $\forall x_1\ldots \forall x_n\psi$, where $\psi$ contains no propositional variables and the only state symbols from $\psi$ are constant nominals or state variables  from  $\{x_1,\ldots, x_n\}$. 
\end{definition}

\begin{proposition}[Pure formulas in ${\mathcal H}_{\bf \Sigma}(@,\forall)$]\label{prop:pure2}
Let ${\mathcal M}$ be a named model where ${\mathcal M}=(W, (R_\sigma)_{\sigma\in\Sigma}, (w^c)_{c\in {\nom}}, V)$, ${\mathcal F}=(W, (R_\sigma)_{\sigma\in\Sigma}, (w^c)_{c\in  {\nom}})$ the corresponding frame and $\phi$ a $\forall$-pure formula of sort $s$. Then ${\mathcal F}\mos{s}\phi$ if and only if ${\mathcal M}\mos{s}\phi$.
\end{proposition}

\begin{proof}
Assume that  ${\mathcal M}\mos{s}\forall x_1 \ldots \forall x_n  \psi$ and ${\mathcal F}\not \mos{s}\forall x_1 \ldots \forall x_n  \psi$  Hence, for some model ${\mathcal M}'$, assignment $g'$ and some $w\in W_{s}$ of sort $s$, ${\mathcal M}', g', w\not \mos{s}\psi$. Since the only state symbols from $\psi$ are constant nominals or state variables  from  $\{x_1,\ldots, x_n\}$, we get ${\mathcal M}, g', w\not \mos{s}\psi$.  But this contradicts the hypothesis  ${\mathcal M}\mos{s}\forall x_1 \ldots \forall x_n  \psi$. In conclusion, ${\mathcal F}\not \mos{s}\forall x_1 \ldots \forall x_n  \psi$.

%
\end{proof}

We are ready now to define the deductive systems of our logics.
The deductive systems for  ${\mathcal H}_{\bf \Sigma}(@)$ and ${\mathcal H}_{\bf \Sigma}(@,\forall)$ are presented in Figure \ref{fig:unu}.

\begin{figure}[!ht]
\centering

{\bf The system} ${\mathcal H}_{\bf\Sigma}(@,\forall)$
\begin{itemize}
\item The axioms and the deduction rules of ${\mathcal K}_{\bf\Sigma}$

\item Axiom schemes: any formula of the following form is an axiom, where $s,s',t$ are sorts,  $\sigma\in\Sigma_{s_1\cdots s_n,s}$,  $\phi,\psi, \phi_1,\ldots,\phi_n$ are formulas (when necessary, their sort is marked as a subscript), $j,k$ are nominals or constant nominals, and $x$, $y$ are state variables:

$\begin{array}{rlrl}
(K@) & @_j^{s} (\phi_t \to \psi_t) \to (@_j^s \phi \to @_j^s \psi) &
(Agree) &  @_k^{t}@_j^{t'} \phi_s \leftrightarrow @^t_j \phi_s\\
(SelfDual) & @^s_j \phi_t \leftrightarrow \neg @_j^s \neg \phi_t &
(Intro)  & j \to (\phi_s \leftrightarrow @_j^s \phi_s)\\
(Back) & \sigma(\ldots,\phi_{i-1}, @_j^{s_i} {\psi}_t,\phi_{i+1},\ldots)_s\to @_j^s {\psi}_t & 
(Ref) & @_j^sj_t 
\end{array}$\\
\medskip

$\begin{array}{rl}
(Q1) & \forall x\,(\phi \to \psi) \to (\phi \to \forall x\, \psi)
  \mbox{ where $\phi$ contains no free occurrences of x}\\
(Q2) & \forall x\,\phi \to \phi[y\slash x] 
  \mbox{ where $y$ is substitutable for $x$ in $\phi$}\\
(Name) & \exists x \, x \\ 
(Barcan) & \forall x \,\sigma^{\mb}(\phi_1, \ldots, \phi_n) \to \sigma^{\mb}(\phi_1, \ldots,\forall x\phi_{i},\ldots, \phi_n) \mbox{ if $x$ does not}\\
& \hfill \mbox{appear free in $\phi_k$ for any $k\neq i$}\\
(Barcan@) & \forall x \, @_j\phi \to @_j \forall x\,\phi\\
 (Nom\, x) & @_k x\wedge @_j x \to @_k j 
\end{array}$

\medskip

\item Deduction rules:

\begin{tabular}{rl}
$(BroadcastS)$ & if $\vds{s}@_j^s\phi_t$ then $\vds{s'}@_j^{s'}\phi_t$ \\
 $(Gen@)$& if $\vds{s'} \phi$ then $\vds{s} @_j \phi$, where $j $ and $\phi$ have the same sort $s'$\\
$(Name@)$ & if $\vds{s} @_j \phi$ then $\vds{s'} \phi$, where $j$ does not occur in $\phi$ and $j\not\in N$\\
$(Paste)$ & if $\vds{s} @_j \sigma(\ldots, k,\ldots) \wedge @_k \phi \to \psi$ then $\vds{s} @_j \sigma(\ldots, \phi, \ldots) \to \psi$\\ & where $k\not\in N$ is distinct from $j$ that does not occur in any\\ 
& other formula that appears in the hypothesis of the deduction rule\\
$(Gen)$ & if $\vds{s}\phi$ then $\vds{s}\forall x\phi$\\
&  where $\phi\in Form_s$ and $x\in SVAR_t$ for some $t\in S$.
\end{tabular}

Here, $j$ and $k$ are nominals or constant nominals having the appropriate sort. Note that in $(Name@)$ and $(Paste)$ we require that some nominals are not constant. 
\end{itemize}

\medskip

{\bf The system ${\mathcal H}_{\bf \Sigma}(@)$}
\begin{itemize}
\item The axioms and the deduction rules of ${\mathcal K}_{\bf\Sigma}$
\item Axiom schemes: $(K@),(SelfDual),(Back),(Agree),(Intro),(Ref)$
\item Deduction rules: $(BroadcastS)$, $(Gen@)$, $(Subst)$, $(Name@)$, $(Paste)$  
\end{itemize}

\caption{$(S,\Sigma)$ hybrid logic}\label{fig:unu}
\end{figure}

In the sequel, our main focus is on the more expressive system ${\mathcal H}(@,\forall)$. The properties and the proofs for 
 ${\mathcal H}(@)$ follow easily from their equivalent in the richer setting. 
 
Theorems and (local) deduction from hypothesis are defined as in Section \ref{sec:prel}. In order to further develop our framework, we need to analyze the {\em uniform substitutions}. Apart for being $S$-sorted, in the  hybrid setting, more restrictions are required: state variables are uniformly replaced by state symbols that are {\em substitutable} for them (as in the mono-sorted setting \cite{hyb}). Nominals and constant nominals are always substitutable for state variables of the same sort.  If $x$ and $z$ are state variables of the sort $s$, then we define:
\begin{itemize}
\item  if $\phi \in {\rm PROP}_s\cup {\rm SVAR}_{s}\cup {\rm NOM}_s \cup \nom_s$, then $z$ is substitutable for $x$ in $\phi$,
\item $z$ is substitutable for $x$ in $\neg \phi$ iff $z$ is substitutable for $x$ in $ \phi$,
\item $z$ is substitutable for $x$ in $\phi \vee \psi$ iff $z$ is substitutable for $x$ in $ \phi$ and $\psi$,
\item $z$ is substitutable for $x$ in $\sigma(\phi_1, \ldots, \phi_n)$ iff $z$ is substitutable for $x$ in $ \phi_i$ for all $i\in [n]$,
\item $z$ is substitutable for $x$ in $@_j^s \phi$ iff $z$ is substitutable for $x$ in $ \phi$,
\item $z$ is substitutable for $x$ in $\forall y\, \phi$ iff $x$ does not occur free in $\phi$, or $y \neq z$ and  $z$ is substitutable for $x$ in $\phi$.
\end{itemize}
\noindent In the sequel, we will say that a substitution  is {\em legal} if it perform only allowed replacements.  If $\phi$ is a formula and $x$ is a state variable we denote by $\phi[z/x]$ the formula obtained by substituting $z$ for all free occurrences of 
$x$ in $\phi$ ($z$ must be a nominal, a constant nominal or a state variable substitutable for $x$).  

\begin{lemma}[Agreement Lemma]
\label{lem:agree} 
Let $\mathcal{M}$ be a standard model. For all standard
$\mathcal{M}$-assignments $g$ and $h$, all states $w$ in $\mathcal{M}$ and all formulas $\phi$ of sort $s \in S$, if $g$ and $h$ agree on all state variables occurring freely in $\phi$, then:
$$\mathcal{M},g,w \mos{s} \phi \ \mbox{iff} \ \mathcal{M},h,w \mos{s} \phi$$ 
\end{lemma}
\begin{proof}
We suppose that $g$ and $h$ agree on all state variables occurring freely in $\phi$ on each sort. We prove this lemma by induction on the complexity of $\phi$:
\begin{itemize}
\item $\mathcal{M},g,w \mos{s} a$ iff $a\in {\rm PROP_s}\cup {\rm NOM_s}\cup  {\nom}_s$ we have $w\in V_s^{\nom}(a)$ iff $\mathcal{M},h,w \mos{s} a$.

\item $\mathcal{M},g,w \mos{s} x$ iff $x \in {\rm SVAR_s}$ we have $ w=g_s(x)$, but $g_s(x)=h_s(x)$, therefore $\mathcal{M},h,w \mos{s} x$.

\item $\mathcal{M},g,w \mos{s} \neg \phi$ iff $\mathcal{M},g,w \mosn{s}\phi$. But, if $g$ and $h$ agree on all state variables occurring freely in $\neg \phi$, then same for $\phi$. Therefore, from the induction hypothesis, $\mathcal{M},g,w \mos{s} \phi$ iff $\mathcal{M},h,w \mos{s} \phi$. Then $\mathcal{M},g,w \mosn{s}  \phi$ iff $\mathcal{M},h,w \mosn{s} \phi$. Then $\mathcal{M},h,w \mos{s} \neg \phi$.

\item $\mathcal{M},g,w \mos{s} \phi \vee \psi$, iff $\mathcal{M},g,w \mos{s} \phi$ or $\mathcal{M},g,w \mos{s} \psi$. But, $g$ and $h$ agree on all state variables occurring freely in $ \phi$ or $\psi$, then from induction hypothesis, we have ($\mathcal{M},g,w \mos{s} \phi$ iff $\mathcal{M},h,w \mos{s} \phi$) or ($\mathcal{M},g,w \mos{s} \psi$ iff $\mathcal{M},h,w \mos{s} \psi$). Then, ($\mathcal{M},h,w \mos{s} \psi$ or $\mathcal{M},h,w \mos{s} \psi$) iff $\mathcal{M},h,w \mos{s} \phi \vee \psi$.

\item $\mathcal{M},g,w \mos{s} \sigma (\phi_1, \ldots , \phi_n )$ iff there is   $(w_1,\ldots,w_n) \in W_{s_1}\times\cdots\times W_{s_n}$  such that $R_{\sigma} ww_1\ldots w_n$ and $\mathcal{M},g,w_i  \mos{s_i} \phi_i$ for each $i \in [n]$, then, by induction hypothesis $\mathcal{M},h,w_i  \mos{s_i} \phi_i$ for each $i \in [n]$. Hence, we have that there is   $(w_1,\ldots,w_n) \in W_{s_1}\times\cdots\times W_{s_n}$  such that $R_{\sigma} ww_1\ldots w_n$ and  $\mathcal{M},h,w_i  \mos{s_i} \phi_i$ for each $i \in [n]$ iff $\mathcal{M},h,w \mos{s} \sigma (\phi_1, \ldots , \phi_n )$.

\item $\mathcal{M},g,w \mos{s} @_j^s \phi$ iff $\mathcal{M},g,v \mos{s'} \phi$ where $V_{s'}^{\nom}(j)=\{ v \}$ iff $\mathcal{M},h,v \mos{s'} \phi$ where $V_{s'}^{\nom}(j)=\{ v \}$ (induction hypothesis) iff $\mathcal{M},h,w \mos{s} @_j^s \phi$.

\item $\mathcal{M},g,w \mos{s} \forall x\phi$ iff $\forall g'( g' \stackrel{x}{\sim} g $ implies $ \mathcal{M},g',w \mos{s} \phi)$. But $ g$ and $h$ agree on all state variables occurring freely in $  \forall x\phi$ and because $x$ is bounded, then $h_s(y)=g_s(y)$ for any $y \neq x$. Therefore, $\forall g'( g_s'(y)=g_s(y)=h_s(y)$ for any $y\neq x $ implies  $  \mathcal{M},g',w \mos{s} \phi)$ equivalent with $\forall g'( g' \stackrel{x}{\sim} h$ implies $  \mathcal{M},h',w \mos{s} \phi)$ iff $\mathcal{M},h,w \mos{s} \forall x\phi$.
\hfill$\Box$
\end{itemize}
\end{proof}

\begin{lemma}[Substitution Lemma]\label{lem:subst}
Let $\mathcal{M}$ be a standard model. For all standard
$\mathcal{M}$-assignments $g$, all states $w$ in $\mathcal{M}$ and all formulas $\phi$, if $y$ is a state variable that is substitutable for $x$ in $\phi$ and $j$ is a nominal then:
\begin{itemize}
\item \label{hunu}$\mathcal{M},g,w \mos{s} \phi[y/x]$ iff $\mathcal{M},g',w \mos{s} \phi$ where $g' \stackrel{x}{\sim} g$ and $ g'_s(x) =g_s(y)$
\item \label{hdoi} $\mathcal{M},g,w \mos{s} \phi[j/x]$ iff $\mathcal{M},g',w \mos{s} \phi$ where $g' \stackrel{x}{\sim} g$ and $ g'_s(x) =V_s^{\nom}(j)$
\end{itemize}
\end{lemma}
\begin{proof}
By induction on the complexity of $\phi$.
\begin{itemize}
\item $\phi = a$, $a\in {\rm PROP}_s\cup {\rm NOM}_s\cup {\nom}_s$. Then $a[y/x]=a$ and $\mathcal{M}, g, w \mos{s} a[y/x] $ if and only if $\mathcal{M}, g, w \mos{s} a  $ if and only if $w \in V_s^{\nom}(a)$. But $ g' \stackrel{x}{\sim} g $ and by Agreement Lemma $\mathcal{M}, g', w \mos{s} a  $.
\item $\phi= z$, where $z \in {\rm SVAR}_s$. We have two cases:
\begin{enumerate}
\item If $z\neq x$, then $\mathcal{M}, g, w \mos{s} z[y/x] $ if and only if $\mathcal{M}, g, w \mos{s} z  $ if and only if $\mathcal{M}, g', w \mos{s} z  $ (Agreement Lemma).
\item If $z=x$, then $\mathcal{M}, g, w \mos{s} z[y/x] $ if and only if $\mathcal{M}, g, w \mos{s} y $ if and only if $w \in g_s(y)$ if and only if $w \in g'_s(x)$ if and only if $w \in g'_s(z)$ if and only if $\mathcal{M}, g', w \mos{s} z$.
\end{enumerate}
\item $\phi= \neg \phi$, then $\mathcal{M}, g, w \mos{s} \neg \phi$ if and only if $\mathcal{M}, g, w  \mosn{s} \phi$ if and only if $\mathcal{M}, g', w  \mosn{s} \phi$ (inductive hypothesis) if and only if $\mathcal{M}, g', w \mos{s} \neg \phi$. 
\item $\phi = \phi \vee \psi$, then $\mathcal{M}, g, w \mos{s} (\phi \vee \psi)[y/x] $ if and only if $\mathcal{M}, g, w \mos{s} \phi [y/x] $ or $\mathcal{M}, g, w \mos{s}  \psi[y/x] $ if and only if $\mathcal{M}, g', w \mos{s} \phi $ or $\mathcal{M}, g', w \mos{s}  \psi $ (inductive hypothesis) if and only if $\mathcal{M}, g', w \mos{s}  \phi \vee \psi $.
\item $\phi = \sigma(\phi_1, \ldots, \phi_n)$, then $\mathcal{M}, g, w \mos{s} \sigma(\phi_1, \ldots, \phi_n)[y/x]$ if and only if $\mathcal{M}, g, w \mos{s} \sigma(\phi_1[y/x], \ldots, \phi_n[y/x])$ if and only if exists $(u_1, \ldots, u_n) \in W_{s_1}\times \ldots \times W_{s_n}$ such that $R_{\sigma}wu_1 \ldots u_n$ and  $\mathcal{M}, g, u_i \mos{s_i} \phi_i[y/x]$ for any $i \in [n]$ if and only if there exists $(u_1, \ldots, u_n) \in W_{s_1}\times \ldots \times W_{s_n}$ such that $R_{\sigma}wu_1 \ldots u_n$ and $\mathcal{M}, g', u_i \mos{s_i} \phi_i $ for any $i \in [n]$ (inductive hypothesis) if and only if $\mathcal{M}, g', w \mos{s} \sigma(\phi_1, \ldots, \phi_n)$.
\item $\phi = @_j^s \phi$, then $\mathcal{M}, g, w \mos{s} @_j^s \phi[y/x]$ if and only if $\mathcal{M}, g, v \mos{s} \phi[y/x]$ where $ V^{\nom}_{s'}=\{v\}$ if and only if $\mathcal{M}, g', v \mos{s'} \phi$ where $ V^{\nom}_{s'}=\{v\}$ (inductive hypothesis) if and only if $\mathcal{M}, g', w \mos{s}  @_j^s \phi$.
\item $\phi = \forall x \phi$, then $\mathcal{M}, g, w \mos{s} (\forall x \phi)[y/z]$ if and only if $\mathcal{M}, g, w \mos{s} (\forall x \phi)[y/z]$ if and only if $\mathcal{M}, g, w \mos{s} \forall x \phi$ if and only if $\mathcal{M}, g', w \mos{s} \forall x \phi$ (Agreement Lemma).

For the next case we will use the notation $g^{x \leftarrow y}$ to specify that $x$ is substituted by $y$, therefore, if $x$ if free in a formula, after substitution we will not have any more $x$.

\begin{claim}
The following two statements are equivalent:
\begin{itemize}
\item For all $g'$, if $g'\stackrel{z}{\sim}g$ then $\mathcal{M}, g'^{x\leftarrow y}, w \mos{s} \phi$.
\item For all $g'$, if $g' \stackrel{z}{\sim} g^{x\leftarrow y}$ then $\mathcal{M}, g', w \mos{s} \phi$.
\end{itemize}
\end{claim}
\begin{proof}
Suppose for all $g'$, if $g'\stackrel{z}{\sim}g$ then $\mathcal{M}, g'^{x\leftarrow y}, w \mos{s} \phi$ and $g' \stackrel{z}{\sim} g^{x\leftarrow y}$. Since $g_s'(o)=g_s^{x\leftarrow y}(o)$ for any $o \neq z$ and $x\neq z$, then $g_s'(x)=g_s^{x\leftarrow y}(x)=g_s(y)$. Therefore, $g_s'={g_s'}^{x\leftarrow y}$ and $g'={g'}^{x\leftarrow y}$. Hence, $\mathcal{M}, g', w \mos{s} \phi$. Next, suppose for all $g'$, if $g' \stackrel{z}{\sim} g^{x\leftarrow y}$ then $\mathcal{M}, g', w \mos{s} \phi$ and $g'\stackrel{z}{\sim}g$. Therefore, $g_s'^{x\leftarrow y} \stackrel{z}{\sim} g_s^{x\leftarrow y}$, so $g'^{x\leftarrow y} \stackrel{z}{\sim} g^{x\leftarrow y}$. From second case, we have that $\mathcal{M}, g'^{x\leftarrow y}, w \mos{s} \phi$. 
\end{proof}
\item $\phi = \forall z \phi$, where $z\neq x$. Suppose $\mathcal{M}, g, w \mos{s}(\forall z \phi)[y/x]$ iff $\mathcal{M}, g, w \mos{s}\forall z (\phi[y/x])$ iff for all $g'$, if $g'\stackrel{z}{\sim}g$ then $\mathcal{M}, g', w \mos{s} \phi[y/x]$ iff for all $g'$, if $g'\stackrel{z}{\sim}g$ then $\mathcal{M}, g'^{x\leftarrow y}, w \mos{s} \phi$ (induction hypothesis) iff or all $g'$, if $g' \stackrel{z}{\sim} g^{x\leftarrow y}$ then $\mathcal{M}, g', w \mos{s} \phi$ (Claim 1) iff $\mathcal{M}, g^{x\leftarrow y},w \mos{s} \forall z \phi$ where $g_s'(x)=g(y)$ and  $g' \stackrel{z}{\sim} g$ iff $\mathcal{M}, g',w \mos{s} \forall z \phi$ where $g_s'(x)=g_s(y)$ and  $g' \stackrel{z}{\sim} g$ (Agreement Lemma).
\end{itemize}
\end{proof}

\begin{lemma}[Generalization on nominals]\label{lem:gennom}
Assume \mbox{$\vds{s}\phi[i/x]$} where $i\in {\rm NOM}_t$ and $x\in {\rm SVAR}_t$ for some $t\in S$. Then there is a state variable $y\in {\rm SVAR}_t$ that does not appear in $\phi$ such that $\vds{s}\phi[y/x]$
\end{lemma}

\begin{proof}
There are two cases. First, let us suppose that $x$ does not occur free in $\phi$, therefore $\phi[j/x]$ is identical to $\phi[y/x]$, hence as $\phi[j/x]$ is provable, so is $\forall y \phi [y/x]$ for any choice of $y$. 

Secondly, suppose that $x$ occur free in $\phi$. Suppose  $\phi[j/x]$. Hence we have a proof of $\phi[j/x]$ and we choose any variable $y$ that does not occur in the proof, or in $\phi$. We replace every occurrence of $j$ in the proof of $\phi[j/x]$ with $y$. It follows by induction on the length of proofs that this new sequence is a proof of $\phi[y/x]$. By generalization we extend the proof with $\forall y(\phi[y/x])$ and we can conclude that $\forall y(\phi[y/x])$ is provable.
\end{proof}

The systems ${\mathcal H}_{\bf \Sigma}(@)$ and ${\mathcal H}_{\bf \Sigma}(@,\forall)$ are sound with respect to the intended semantics.

\begin{proposition}[Soundness]\label{prop:sound}
The deductive systems for ${\mathcal H}_{\bf \Sigma}(@)$ and ${\mathcal H}_{\bf \Sigma}(@,\forall)$ from Figure 2 are sound.
\end{proposition}

\begin{proof}
We will only prove the soundness of the more complex system  ${\mathcal H}(@,\forall)$, since this proof is similar for the ${\mathcal H}_{\bf \Sigma}(@)$ system.

Let $\mathcal{M}$ be an arbitrary model and $w$ any state of sort $s$.

$(K_@)$
Suppose $\mathcal{M},g,w \mos{s} @_j^{s} (\phi_t \to \psi_t) $ if and only if $\mathcal{M},g,v \mos{t} \phi_t \to \psi_t$ where $V^{\nom}_t(j)=\{ v\}$ iff  $\mathcal{M},g, v \mos{t} \phi_t$ implies $\mathcal{M},g, v \mos{t} \psi_t$ where $V^{\nom}_t(j)=\{ v\}$. Suppose $\mathcal{M},g, w  \mos{s} @_j^s \phi_t$ and $V^{\nom}_t(j)=\{ v\}$. Then $\mathcal{M},g, v  \mos{t} \phi_t$ where $V^{\nom}_t(j)=\{ v\}$ , but this implies that $\mathcal{M},g, v \mos{t} \psi_t$ where $V^{\nom}_t(j)=\{ v\}$ iff $\mathcal{M},g, w  \mos{s} @_j^s \psi_t$.

$(Agree)$
Suppose $\mathcal{M},g,w  \mos{t'} @_k^{t'}@_j^{t} \phi_s $ iff $\mathcal{M},g,v  \mos{t} @_j^{t} \phi_s $ where $V^{\nom}_t(k)=\{ v\}$ iff $\mathcal{M},g,u  \mos{s}  \phi_s $ where $V^{\nom}_t(k)=\{ v\}$ and $V^{\nom}_s(j)=\{ u\}$. Then $\mathcal{M},g,u  \mos{s}  \phi_s $ where $V^{\nom}_s(j)=\{ u\}$ which implies that $\mathcal{M},g,w  \mos{t'}  @^{t'}_j \phi_s $.

$(SelfDual)$ 
Suppose $\mathcal{M},g,w  \mos{s} \neg @^s_j  \neg \phi_t $ iff $\mathcal{M},g,w  \mosn{s} @^s_j \neg  \phi_t$ iff $\mathcal{M},g,v  \mosn{t}   \neg \phi_t$ where $V^{\nom}_t(j)=\{ v\}$ iff  $\mathcal{M},g, v \mos{t} \phi_t$ where $V^{\nom}_t(j)=\{ v\}$ iff $\mathcal{M},g,w  \mos{s}  @^s_j  \phi_t$.

$(Back)$
Suppose $\mathcal{M},g,w \mos{s}  \sigma(\ldots,\phi_{i-1}, @_j^{s_i} {\psi}_t,\phi_{i+1},\ldots)_s$ if and only if there is $(w_1,\ldots,w_n) \in W_{s_1}\times\cdots\times W_{s_n}$ such that  $R_{\sigma} ww_1\ldots w_n$ and $\mathcal{M},g,w_i  \mos{s_i} \phi_i$ for any $i \in [n]$. This implies that there is $w_i \in W_{s_i}$ such that  $\mathcal{M},g,w_i  \mos{s_i} @_j^{s_i} {\psi}_t $, then $ \mathcal{M},g,v  \mos{t} \psi_t $ where $V^{\nom}_t(j)=\{ v\}$. Hence, $\mathcal{M},g,w \mos{s}  @_j^{s} {\psi}_t $

$(Ref)$ 
Suppose $\mathcal{M},g, w \mosn{s}@_j^sj_t$. Then $\mathcal{M},g, v \mosn{t} j$ where $V^{\nom}_t(j)=\{ v\}$, contradiction.

$(Intro)$
Suppose $\mathcal{M},g,w \mos{s} j$  and  $\mathcal{M},g,w \mos{s}  \phi_s$. Then $V^{\nom}_s(j)=\{ w\}$ and $\mathcal{M},g,w \mos{s} \phi_s$ implies that $\mathcal{M},g,w \mos{s} @_j^s \phi_s$. Now, suppose $\mathcal{M},g,w \mos{s} j$ and $\mathcal{M},g,w \mos{s} @_j^s \phi_s$. Because, from the first assumption, we have $V^{\nom}_s(j)=\{ w\}$, then, form the second one, we can conclude that $\mathcal{M},g,w\mos{s} \phi_s $.

$(Q1)$
Suppose that $\mathcal{M},g,w \mos{s} \forall x(\phi \to \psi)$ iff $\mathcal{M},g',w \mos{s} \phi\to \psi$ for all $g' \stackrel{x}{\sim} g$. Results that for all $g' \stackrel{x}{\sim} g$ we have $\mathcal{M},g',w \mos{s} \phi$ implies $\mathcal{M},g',w \mos{s} \psi$. But $\phi$ contains no free occurrences of $x$, then for all $g' \stackrel{x}{\sim} g$ we have ($\mathcal{M},g,w \mos{s} \phi$ implies $\mathcal{M},g',w \mos{s} \psi$). Hence,  $\mathcal{M},g,w \mos{s} \phi$ implies that, for all $g' \stackrel{x}{\sim} g$, $\mathcal{M},g',w \mos{s} \psi$. Then, $\mathcal{M},g,w \mos{s} \phi$ implies that  $\mathcal{M},g,w \mos{s} \forall \psi$ iff $\mathcal{M},g,w \mos{s} \phi \to \forall x \psi$ .

$(Q2)$ 
Suppose that $\mathcal{M},g,w \mos{s} \forall x\phi$. We need to prove that $\mathcal{M},g',w \mos{s}  \phi[y/x]$. But this is equivalent, by Substitution Lemma, with proving that  $\mathcal{M},g',w \mos{s} \phi$ where $g' \stackrel{x}{\sim} g$ and $ g'_s(x) =g_s(y)$. But $\mathcal{M},g,w \mos{s} \forall x\phi$ iff $ \mathcal{M},g',w \mos{s} \phi$ for all $g' \stackrel{x}{\sim} g$. Let $g'_s(z)=g(y)$, if $z=x$, and $g'_s(z)=g(z)$, otherwise. Therefore, we have  $g' \stackrel{x}{\sim} g$ , $ g'_s(x) =g_s(y)$ and $ \mathcal{M},g',w \mos{s} \phi$. For the case of substituting with a nominal is similar. We define  $ g'_s(x)=V^{\nom}_s(j)$, if $z=x$, and $g'_s(z)=g(z)$, otherwise. 

$(Name)$ Suppose that $\mathcal{M},g,w \mos{s} \exists x x$ iff exists $g' \stackrel{x}{\sim} g$ and $ \mathcal{M},g',w \mos{s} x$. We choose $g'$ an $x$-variant of $g$ such that $g'_s(x)= \lbrace w\rbrace$.

$(Barcan)$ Assume  $g$ is an assignment and $R_{\sigma} ww_1\ldots w_n$. We have to prove that $\mathcal{M},g,w_i \mos{s_i}\forall x \phi_i$ or $\mathcal{M},g,w_k \mos{s_k}\phi_k$ for some $k\neq i$.
By hypothesis,  $\mathcal{M},g,w \mos{s} \forall x \, \sigma^{\mb}(\phi_1, \ldots, \phi_n)$, so for all $g' \stackrel{x}{\sim} g$,
$\mathcal{M},g',w \mos{s} \sigma^{\mb}(\phi_1, \ldots, \phi_n)$.
This means that, for all $g' \stackrel{x}{\sim} g$, we have $\mathcal{M},g',w_i \mos{s_i} \phi_i$ or $\mathcal{M},g',w_k \mos{s_k}\phi_k$ for some $k\neq i$. We consider two cases:

- there is an  assignment $g' \stackrel{x}{\sim} g$ such that  $\mathcal{M},g',w_k \mos{s_k}\phi_k$ for some $k\neq i$; since $x$ does not appear in $\phi_k$, we infer $\mathcal{M},g,w_k \mos{s_k}\phi_k$;

-$\mathcal{M},g',w_i \mos{s_i} \phi_i$ for any assignment $g' \stackrel{x}{\sim} g$; but this implies $\mathcal{M},g,w_i \mos{s_i} \forall x\phi_i$.
 
\noindent We proved that $\mathcal{M},g,w_i \mos{s_i}\forall x \phi_i$ or $\mathcal{M},g,w_k \mos{s_k}\phi_k$ for some $k\neq i$, which is the desired conclusion. 


$(Barcan@)$
Suppose $\mathcal{M},g,w \mos{s} \forall x @_j^s \phi$ iff $\mathcal{M},g',w \mos{s}  @_j^s \phi$ for all $g' \stackrel{x}{\sim} g$. Then, $\mathcal{M},g',v \mos{t}  \phi$ for all $g' \stackrel{x}{\sim} g$ where $V^{\nom}_t(j)=\{ v\}$ and so $\mathcal{M},g,v \mos{t} \forall x \phi$ where $V^{\nom}_t(j)=\{ v\}$. Hence, $\mathcal{M},g,w \mos{s}  @_j^s \forall x \phi$.

$(Nom\ x)$ 
Suppose $\mathcal{M},g,w \mos{s}  @_j^s  x$ and $\mathcal{M},g,w \mos{s}  @_k^s  x$. Then $\mathcal{M},g,v \mos{t}  x$ where $V^{\nom}_t(j)=\{ v\}$ and  $\mathcal{M},g,u \mos{t}  x$ where $V^{\nom}_t(k)=\{ u\}$. This implies that $u=v$, so  $V^{\nom}_t(j)=V^{\nom}_t(k)$. Then  $\mathcal{M},g,w \mos{s} @_j^s k$ for any model $\mathcal{M}$ and any world $w$.

$(BroadcastS)$ 
Suppose $\mathcal{M},g,w \mos{s} @_j^s  \phi_t$  if and only if $\mathcal{M},g,v \mos{t} \phi_t$ where $V^{\nom}_t(j)=\{ v\}$. Hence, for any $s' \in S$ we have $\mathcal{M},g,w \mos{s'} @_j^{s'}  \phi_t$.

Now, let $\mathcal{M}$ be an arbitrary named model.

$(Name @)$ Let $\mathcal{M}$ be a model, $g$ an assignment and $v\in W_{s'}$ and assume that $\mathcal{M},g,v \not\mos{s'}  \phi$, where $s'$ is the sort of $\phi$. Since  $j\in {\rm NOM}_{s'}$ and $j$ does not appear in $\phi$, we can safely assume that $g_{s'}(j)=\{v\}$. Hence  $\mathcal{M},g,w \not\mos{s} @_j^s  \phi$ for any $w\in W_s$, which contradicts the hypothesis.  


$(Paste)$ Suppose $\mathcal{M},g,w \mos{s} @_j^s  \sigma(\psi_1, \ldots, \psi_{i-1}, k, \psi_{i+1}, \ldots, \psi_n) \wedge @_k^s \phi \to \psi$ and $\mathcal{M},g,w \mos{s} @_j^s  \sigma(\psi_1, \ldots, \psi_{i-1}, \phi, \psi_{i+1}, \ldots, \psi_n$. Hence there exists $(v_1, \ldots,v_n) \in (W_{s_1}\times \ldots\times W_{s_n})$ such that $R_{\sigma}v v_1 \ldots v_i \ldots v_n$,  $V^{\nom}_{s'}(j)=\{v \}$ and  $\mathcal{M},g,v_e \mos{s'}  \psi_e$ for any $e \in [n], e\neq i$ and $  \mathcal{M}, g, v_i \mos{s_i} \phi$. Let $k\in {\rm{NOM}_{s_i}}$ such that $k$ does not appear in other formulas as required. Then we can assume  consider that $g(k)=\{v_i\}$ so the hypothesis of the deduction rule is satisfied and we can infer that $\mathcal{M},g,w \mos{s}\psi$.

%

In conclusion, $\mathcal{M},g,w \mos{s'}  @_j^s\sigma(\psi_1, \ldots, \psi_{i-1}, \phi, \psi_{i+1}, \ldots, \psi_n) \to \psi$.

\end{proof}

\begin{lemma}\label{lemma:bridge}
\begin{itemize}

\item[1.] The following formulas are theorems:

\begin{tabular}{ll}
$(Nom)$ & $@_k^sj\to (@_k^s\phi\leftrightarrow @_j^s\phi)$\\
& for any $t\in S$,  $k,j\in{\rm NOM}_t\cup{\nom}_t$ and $\phi$ a formula of sort $t$. \\ 
$(Sym)$ &  $@_k^sj\to @_j^sk$ \\
  & where $s\in S$ and $j,k\in {\rm NOM}_t\cup{\nom}_t$ for some $t\in S$,\\
$(Bridge)$ & $\sigma(\ldots \phi_{i_1}, j ,\phi_{i+1} \ldots) \wedge @_j^s \phi \to \sigma(\ldots\phi_{i-1}, \phi ,\phi_{i+1}, \ldots)$ \\
 & if $\sigma\in \Sigma_{s_1\ldots s_n,s}$, 
  $j\in {\rm NOM}_{s_i}\cup{\nom}_{s_i}$ and $\phi$ is a formula of sort $s_i$.
 \end{tabular}
 \medskip
\item[2.] if $\vds{s}\phi\to j$ then 
$\vds{t}\sigma(\ldots,\phi,\ldots)\to \sigma(\ldots,j,\ldots)\wedge @_j^t\phi$

for any $s,t\in S$, $\sigma\in\Sigma_{t_1\cdots t_n,t}$, $j\in {\rm NOM}_s\cup{\nom}_s$ and $\phi$ a formula of sort $s$. 
\end{itemize}
\end{lemma}

\begin{proof}
In the sequel, by PL we mean classical propositional logic and by ML we mean the basic modal logic.

1. $(Nom)$ 
 
\noindent$(1)$ $ \vds{t} j \to (\phi \leftrightarrow @_j^t \phi) $\hfill $(Intro)$

\noindent$(2)$ $ \vds{s} @_k^s (j \to (\phi \leftrightarrow @_j^t \phi))$ \hfill  $(Gen@)$

\noindent$(3)$ $ \vds{s} @_k^s (j \to (\phi \leftrightarrow @_j^t \phi)) \to (  @_k^s j \to  @_k^s (\phi \leftrightarrow @_j^t \phi)) $ \hfill  $(K@)$

\noindent$(4)$  $\vds{s}  @_k^s j \to  @_k^s (\phi \leftrightarrow @_j^t \phi) $ \hfill  $(MP):(2),(3)$

\noindent$(5)$  $\vds{s}  @_k^s (\phi \leftrightarrow @_j^t \phi) \leftrightarrow (@_k^s \phi \leftrightarrow @_k^s @_j^t \phi)$  \hfill ML

\noindent$(6)$  $\vds{s}  @_k^s j \to (@_k^s \phi \leftrightarrow @_k^s @_j^t \phi)$\hfill PL:$(4),(5)$

\noindent$(7)$  $\vds{s}  @_k^s @_j^t \phi  \leftrightarrow  @_j^s \phi$ \hfill $(Agree)$

\noindent$(8)$  $\vds{s}  @_k^s j \to (@_k^s \phi \leftrightarrow @_j^s \phi)$\hfill PL:$(6),(7)$

\medskip

$(Sym)$

\noindent$(1)$ $ \vds{s} @_k^s j \wedge @_j^s k\to @_j^s k$ \hfill $Taut$\\
\noindent$(2)$ $ \vds{s} @_k^s j \wedge @_j^s k \to @_j^s k) \to (@_k^s j \to( @_j^s k\to @_j^s k))$ \hfill $Taut$\\
\noindent$(3)$ $ \vds{s} @_k^s j \to( @_j^sk\to @_j^s k)$ \hfill $(MP):(1),(2)$\\
\noindent$(4)$ $ \vds{s}( @_j^sk\to @_j^s k)\to @_j^s k$ \hfill  PL\\
\noindent$(5)$ $ \vds{s} @_k^s j \to @_j^s k$  \hfill PL\\
\noindent$(6)$ $ \vds{s} @_j^s k \to @_k^s j$  \hfill Analogue\\
\noindent$(7)$ $ \vds{s} @_j^s k \leftrightarrow @_k^s j$\hfill PL:$(5),(6)$
 
   \medskip

$(Bridge)$ 

\noindent$(1)$  $\  \vds{s} \sigma(\ldots \phi_{i_1}, j ,\phi_{i+1} \ldots) \wedge \sigma^{\mb}(\ldots,\neg\phi_{i-1}, \phi ,\neg\phi_{i+1}, \ldots)\to \hfill $

\noindent\hfill$\sigma (\ldots\phi_{i-1},j \wedge \phi ,\phi_{i+1}, \ldots)$ ML

\noindent$(2)$  $\vds{s}j \wedge\phi \to @_j^s \phi $ \hfill $(Intro)$\\
$(3)$  $\vds{s}\sigma (\ldots\phi_{i-1},j \wedge \phi ,\phi_{i+1}, \ldots) \to \sigma (\ldots\phi_{i-1}, @_j^s \phi ,\phi_{i+1}, \ldots)$ \hfill ML

\noindent$(4)$  $\vds{s} \sigma (\ldots\phi_{i-1}, @_j^s \phi ,\phi_{i+1}, \ldots) \to @_j^s \phi$ \hfill  $(Back)$

\noindent$(5)$   $ \vds{s} \sigma(\ldots \phi_{i_1}, j ,\phi_{i+1} \ldots) \wedge \sigma^{\mb}(\ldots,\neg\phi_{i-1}, \phi ,\neg\phi_{i+1}, \ldots)\to  @_j^s \phi$\hfill PL

\noindent$(6)$   $ \vds{s} \sigma(\ldots \phi_{i_1}, j ,\phi_{i+1} \ldots) \wedge \sigma^{\mb}(\ldots,\neg\phi_{i-1}, \neg\phi ,\neg\phi_{i+1}, \ldots)\to  @_j^s \neg\phi$ \hfill $(5)$

\noindent$(7)$   $ \vds{s} \neg @_j^s \neg\phi \to \neg (\sigma(\ldots \phi_{i_1}, j ,\phi_{i+1} \ldots) \wedge \sigma^{\mb}(\ldots,\neg\phi_{i-1},\neg\phi ,\neg\phi_{i+1}, \ldots) )$ \hfill PL

\noindent$(8)$   $ \vds{s}  @_j^s \phi \to (\neg \sigma(\ldots \phi_{i_1}, j ,\phi_{i+1} \ldots) \vee \neg \sigma^{\mb}(\ldots,\neg \phi_{i-1}, \neg\phi ,\neg \phi_{i+1}, \ldots) )$ \hfill PL

\noindent$(9)$   $ \vds{s}  @_j^s \phi \to (\neg \sigma(\ldots \phi_{i_1}, j ,\phi_{i+1} \ldots) \vee \sigma(\ldots,\phi_{i-1}, \phi , \phi_{i+1}, \ldots) )$ \hfill$(Dual)$

\noindent$(9)$   $ \vds{s}  @_j^s \phi \to ( \sigma(\ldots \phi_{i_1}, j ,\phi_{i+1} \ldots) \to \sigma(\ldots,\phi_{i-1}, \phi , \phi_{i+1}, \ldots) )$ \hfill PL

\noindent$(10)$   $ \vds{s}  @_j^s \phi \wedge  \sigma(\ldots \phi_{i_1}, j ,\phi_{i+1} \ldots) \to \sigma(\ldots,\phi_{i-1}, \phi , \phi_{i+1}, \ldots) $ \hfill PL

 \medskip
2.

\noindent$(1)$  $ \vds{s} j \to (\neg \phi \leftrightarrow @^s_j \neg \phi)$ \hfill  $(Intro)$\\
$(2)$   $ \vds{s} j \to (\neg \phi \leftrightarrow @^s_j \neg \phi) \to (j \to ( @_j^s \neg \phi \to \neg \phi))$\hfill  PL\\
$(3)$   $ \vds{s} j \to ( @_j^s \neg \phi \to \neg \phi)$\hfill $(MP):(1),(2)$\\
$(4)$   $ \vds{s} (j \to ( @_j^s \neg \phi \to \neg \phi))\to (j \wedge @_j^s \neg \phi\to \neg \phi) $ \hfill PL\\
$(5)$   $ \vds{s} j \wedge @_j^s \neg \phi\to \neg \phi $\hfill $(MP):(3),(4)$\\
$(6)$   $ \vds{s} \phi\to (\neg j \vee @_j^s \phi)$\hfill PL,$(SelfDual)$\\
$(7)$ $\vds{s}\phi \to j$ \hfill hypothesis\\
$(8)$   $ \vds{s} \phi\to (\neg j \vee @_j^s \phi)\wedge j $\hfill PL\\
$(9)$   $ \vds{s} \phi\to @_j^s \phi\wedge j $\hfill PL\\
$(10)$   $ \vds{s} (\phi\to @_j^s \phi)\wedge(\phi \to j) $\hfill PL\\
$(11)$   $ \vds{s} \phi\to @_j^s \phi $\hfill PL\\

Therefore, if $\vds{s}\phi \to j$ then  $ \vds{s} \phi\to @_j^s \phi $.

\noindent$(1)$ $\vds{s}\phi \to j$ \hfill hypothesis

\noindent$(2)$ $\vds{t}\sigma(\ldots,\psi_{i-1}, \phi, \psi_{i+1},\ldots) \to \sigma(\ldots,\psi_{i-1}, j,\psi_{i+1},\ldots)$\hfill ML$(1)$

\noindent$(3)$   $ \vds{s} \phi\to @_j^s \phi $ \hfill $ (1)$

\noindent$(4)$ $\vds{t}\sigma(\ldots,\psi_{i-1}, \phi, \psi_{i+1},\ldots) \to \sigma(\ldots,\psi_{i-1},  @_j^s \phi,\psi_{i+1},\ldots)$ \hfill ML$(3)$

\noindent$(5)$ $\vds{t}\sigma(\ldots,\psi_{i-1}, \phi, \psi_{i+1},\ldots) \to  @_j^t \phi$ \hfill $(Back)$,PL$(4)$

\noindent$(6)$ $\vds{t}\sigma(\ldots,\psi_{i-1}, \phi, \psi_{i+1},\ldots) \to  (\sigma(\ldots,\psi_{i-1}, j,\psi_{i+1},\ldots) \wedge @_j^t \phi)$\hfill PL:$(2),(5)$

Therefore, if $\vds{s}\phi\to j$ then 
$\vds{t}\sigma(\ldots,\phi,\ldots)\to \sigma(\ldots,j,\ldots)\wedge @_j^t\phi$.
\end{proof}

Let $\bot_s$ denote a formula of sort $s$ that is \textit{nowhere true}. If $s\in S$ and  $\Gamma_s$ is a set of formulas of sort $s$, then $\Gamma_s$ is  \textit{consistent} if $\Gamma_s\ \not\!\!\!\!\vds{s} \bot_s$. An \textit{inconsistent} set of formulas is a set of formulas  of the same sort that is not consistent.  Maximal consistent sets are defined as usual.

\medskip

In the rest of the section we develop the proof of the strong completeness theorem for our hybrid logical systems, possibly extended with additional axioms.  If $\Lambda$ is a set of formulas, we denote by  ${\mathcal H}(@)+\Lambda$ and ${\mathcal H}(@,\forall)+\Lambda$ the systems obtained when the formulas of $\Lambda$ are seen as additional axiom schemes. 

The main steps are: the extended Lindenbaum Lemma, the construction of the Henkin model and the Truth Lemma (all of them extending the similar results in the mono-sorted case). In order to state our extended Lindenbaum Lemma, we need to define the \textit{named, pasted and $@$-witnessed} sets of formulas.

\begin{definition}[Named, pasted and $@$-witnessed sets]
Let $s\in S$ and $\Gamma_s$ be a set of formulas of sort $s$ from 
 ${\mathcal H}_{\bf \Sigma}(@)$. We say that 
 \vspace*{-0.2cm}
 \begin{itemize}
 \item $\Gamma_s$ is {\sf named} if one of its elements is a nominal or a constant nominal,
 \item $\Gamma_s$ is {\sf pasted} if, for any $t\in S$,  $\sigma\in\Sigma_{s_1\cdots s_n,t}$,  $k\in {\rm NOM}_{t}\cup{\nom}_t$, and $\phi$ a formula of sort $s_i$, whenever $@_k^s\sigma(\ldots, \phi_{i-1},\phi,\phi_{i+1},\ldots)\in \Gamma_s$ there exists a nominal $j\in {\rm NOM}_{s_i}$ such that $@_k^s\sigma(\ldots, \phi_{i-1},j,\phi_{i+1},\ldots)\in \Gamma_s$ and  $@_j^s\phi\in \Gamma_s$. 
\end{itemize} 
If $\Gamma_s$ be a set of formulas of sort $s$ from ${\mathcal H}_{\bf \Sigma}(@,\forall)$ then we say that 
 \vspace*{-0.2cm}
 \begin{itemize}
 \item $\Gamma_s$ is {\sf $@$-witnessed} if the following two conditions are satisfied:
\begin{itemize}
 \item[{\rm (-)}] for $s',t\in S$ , $x\in {\rm SVAR}_t$, $k\in {\rm NOM}_{s'}\cup{\nom}_{s'}$ and any formula $\phi$ of sort $s'$,  whenever $@_k^s\exists x\, \phi\in \Gamma_s$ there exists $j\in {\rm NOM}_t$ such that $@_k^s\phi[j/x]\in\Gamma_s$,
 \item[{\rm (-)}] for any $t\in S$ and $x\in {\rm SVAR}_t$ there is $j_s\in {\rm NOM}_t$ such that $@_{j_x}^s x\in \Gamma_s$. 
 \end{itemize}
 \end{itemize}
 \end{definition}

\begin{lemma}[Extended Lindenbaum Lemma]\label{lemma:lind}
Let $\Lambda$ be a set of  formulas in the language of ${\mathcal H}_{\bf\Sigma}(@)$ {\rm (}in the language of ${\mathcal H}_{\bf\Sigma}(@,\forall)${\rm )} and $s\in S$. Then any consistent set $\Gamma_s$ of formulas of sort $s$ from  ${\mathcal H}_{\bf\Sigma}(@)+\Lambda$ {\rm (}from ${\mathcal H}_{\bf\Sigma}(@,\forall)+\Lambda${\rm )} can be extended to a named and pasted {\rm (}named, pasted and $@$-witnessed {\rm )} maximal consistent set by adding countably many nominals to the language.   
\end{lemma}

\begin{proof}
The proof generalizes to the $S$-sorted setting well-known proofs for the mono-sorted hybrid logic, see \cite[Lemma 7.25]{mod}, \cite[Lemma 3, Lemma 4]{pureax}, \cite[Lemma 3.9]{hyb}. 

For each sort $s\in S$, we add a set of new nominals and enumerate this set. Given a set of formulas $\Gamma_s$, define $\Gamma_s^k$ to be $\Gamma_s \cup \{ k_s\} \cup \{@_{j_x}^s x| \ x \in {\rm SVAR_s} \}$, where $k_s$ is the first new nominal of sort $s$ in our enumeration and $j_x$ are such that if $x$ and $y$ are different state variables of sort $s$ then also $j_x$ and $j_y$ are different nominals of same sort $s$. Now that we know we are working on the sort $s$, we will write $
k$ instead of $k_s$.

Suppose $\Gamma_s^k$ is not consistent. Then there exists some conjunction of formulas $\theta \in \Gamma_s$ such that $\vds{s} k \to \neg \theta$. We use the $(Gen@)$ rule and the $(K@)$ axiom to prove that \mbox{$\vds{s} @_k^s k \to @_k^s \neg \theta$}. From the $(Ref)$ axiom and the $(MP)$ rule it follows $\vds{s}@_k^s \neg \theta$. Remember that $k$ is a new nominal, so it does not occur in $\theta$ and we use $(Name@)$ rule to get that $\vds{s} \neg \theta\Rightarrow \neg \theta \in \Gamma_s$. But this contradicts the consistency of $\Gamma_s$. Now, we prove the case for the additional $@_{j_x}^s x$ formulas. Suppose $\vds{s} \theta \to \neg @_{j_x}^s x$. We use the $(SelfDual)$ axiom to get $\vds{s} \neg \theta \vee @_{j_x}^s \neg x$. If \mbox{$\vds{s} \neg \theta $}, this contradicts the consistency of $\Gamma_s$. If \mbox{$\vds{s}  @_{j_x}^s \neg x$}, then $\mos{s} @_{j_x}^s \neg x$. Hence, for any model $\mathcal{M}$, any assignment function $g$ and any world $w \in W_s$, we have $\mathcal{M},g, w \mos{s} @_{j_x}^s \neg x$ if and only if $\mathcal{M}, g ,v\mos{s} \neg x$ where $V^{\nom}_s(j_x)=\{v\}$. Then for any model $\mathcal{M}$ and any assignment $g$, $g_s(x) \neq V^{\nom}_s(j_x)$, contradiction.

Now we enumerate  on each sort $s \in S$  all the formulas of the new language obtained by adding the set of new nominals and define $\Gamma^0 := \Gamma_s^k$. Suppose we have defined $\Gamma^m$, where $m \geq 0$. Let $\phi_{m+1}$ be the $m+1-th$ formula of sort $s$ in the previous enumeration. We define $\Gamma^{m+1}$ as follows. If $\Gamma^{m}\cup \{\phi_{m+1}\}$ is inconsistent, then $\Gamma^{m+1} = \Gamma^{m}$. Otherwise:
\begin{itemize}
\item[(i)] $\Gamma^{m+1} = \Gamma^{m} \cup  \{\phi_{m+1}\} $, if $\phi_{m+1}$ is neither  of the form $@_j\sigma(\ldots, \varphi, \ldots)$, nor of the form $@_j \exists x\varphi(x)$, where $j$ is any nominal of sort $s''$, $\varphi$ a formula of sort $s''$ and $x \in {\rm SVAR_{s''}}$.
\item[(ii)] $\Gamma^{m+1} = \Gamma^{m} \cup  \{\phi_{m+1}\} \cup \{@_j \sigma(\ldots, k, \ldots) \wedge @_k \varphi \} $, if $\phi_{m+1}$ is of the form $@_j \sigma(\ldots, \varphi, \ldots)$.
\item[(iii)] $\Gamma^{m+1} = \Gamma^{m} \cup  \{\phi_{m+1}\} \cup \{ @_j \varphi[k/x]\}$, where $\phi_{m+1} $ is of the form $@_j \exists x\varphi(x)$.
\end{itemize}
In clauses $(ii)$ and $(iii)$, $k$ is the first new nominal in the enumeration that does not occur neither in $\Gamma^i$ for all $i \leq m$, nor in $@_j \sigma(\ldots, \varphi, \ldots)$.

Let $\Gamma ^+= \bigcup_{n\geq 0} \Gamma^n$. Because $k \in \Gamma^0 \subseteq \Gamma^+$, this set in named, maximal, pasted and $@$-witnessed by construction. We will check if it is consistent for the expansion made in the second and third items.

Suppose $\Gamma^{m+1} = \Gamma^{m} \cup  \{\phi_{m+1}\} \cup \{@_j \sigma(\ldots, k, \ldots) \wedge @_k \varphi \} $ is an inconsistent set, where $\phi_{m+1}$ is $@_j \sigma(\ldots, \varphi, \ldots)$. Then there is a conjunction of formulas $\chi \in \Gamma^m \cup \{\phi_{m+1}\} $ such that \mbox{$\vds{s} \chi \to \neg ( @_j \sigma(\ldots, k, \ldots) \wedge @_k \varphi)$} and so \mbox{$\vds{s} @_j \sigma(\ldots, k, \ldots) \wedge @_k \varphi \to \neg \chi$.} But $k$ is the first new nominal in the enumeration that does not occur neither in $\Gamma^m$, nor in $@_j \sigma(\ldots, \varphi, \ldots)$ and by Paste rule we get $\vds{s} @_j \sigma(\ldots, \varphi, \ldots) \to \neg \chi \Rightarrow \vds{s} \chi \to \neg @_j \sigma(\ldots, \varphi, \ldots)$, which contradicts the consistency of $\Gamma^m \cup  \{\phi_{m+1}\}$. 

Suppose  $\Gamma^{m+1} = \Gamma^{m} \cup  \{\phi_{m+1}\} \cup \{ @_j \varphi[k/x]\}$ is inconsistent, where $\phi_{m+1}$ is $ @_j \exists x\varphi(x)$. Then there is a conjunction of formulas $\chi \in \Gamma^m \cup \{\phi_{m+1}\}$ such that \mbox{$\vds{s}$ $ \chi \to \neg  @_j \varphi[k/x]$}, where $k$ is the new nominal. By generalization on nominals (Lemma \ref{lem:gennom}) we can prove \mbox{$\vds{s}$} $ \forall y( \chi \to \neg   @_j \varphi[y/x])$, where $y$ is a state variable that does not occur in $\chi \to \neg   @_j \varphi[k/x]$. Using $(Q1)$ axiom, we get \mbox{$\vds{s}$} $ \chi \to   \forall y\neg   @_j \varphi[y/x]$ and by $(SelfDual)$ \mbox{$\vds{s}$} $ \chi \to   \forall y   @_j \neg\varphi[y/x] $. Next, we use $(Barcan@)$ to get \mbox{$\vds{s}$} $ \chi \to  @_j  \forall y \neg \varphi[y/x])$. Because $x$ has no free occurrences in $\varphi[y/x]$, we can prove that $  @_j  \forall y \neg \varphi[y/x]) \leftrightarrow @_j \forall x \neg \varphi$. Therefore, \mbox{$\vds{s}$ $ \chi \to  @_j \forall x \neg \varphi$ }, so \mbox{$\vds{s}$} $ \chi \to  @_j \neg \exists x \varphi $ . Use once again $(SelfDual)$ and we have \mbox{$\vds{s}$} $ \chi \to \neg @_j  \exists x \varphi $. Then $\neg @_j  \exists x \varphi $ $\in \Gamma^m \cup \{\phi_{m+1}\}$, but this contradicts the consistency of $\Gamma^m \cup \{\phi_{m+1}\}$.
\end{proof}

We are now ready to define a Henkin  model, see \cite{hand,pureax}
for the mono-sorted hybrid modal logic. 

\begin{definition}[The Henkin model]\label{def:model} Let $s\in S$ and assume  $\Gamma_s$ is a maximal consistent set of formulas of sort $s$ from ${\mathcal H}_{\bf\Sigma}(@)$ {\rm(}from ${\mathcal H}_{\bf\Sigma}(@,\forall)${\rm)}. For any $t\in S$ and $j\in {\rm NOM}_t\cup {\nom}_t$ we define 
$|j|=\{k\in NOM_t\cup{\nom}_t\,| \, @_j^sk\in \Gamma_s\}$. The {\sf Henkin model} is  ${\mathcal M}^{\Gamma_s}=(W^\Gamma, (R_\sigma^\Gamma)_{\sigma\in\Sigma},(|c|)_{c\in{\nom}},V^\Gamma)$ where
 \vspace*{-0.2cm}
\begin{center}
\begin{tabular}{rcl}
$W_t^\Gamma$ & $=$ & $\{|j|\,\,| j\in {\rm NOM}_t\cup{\nom}_t\}$ for any $t\in S$\\
$(|j|,|j_1|,\ldots,|j_n|)\in R_\sigma^\Gamma$ & iff  & $@_j^s\sigma(j_1,\ldots,j_n)\in\Gamma_s$ for any $\sigma\in \Sigma_{t_1\cdots t_n,t}$ \\
$V_t^\Gamma(p)$ & $=$ & $\{|j|\,\, | j\in {\rm NOM}_t\cup {\rm {\nom}}_t, @_j^sp\in \Gamma_s\}$\\
 & & \hfill for any $t\in S$ and $p\in {\rm PROP}_t$\\
$V_t^\Gamma(j)$ & $=$ & $\{|j|\}$  for any $t\in S$ and $j\in {\rm NOM}_t$.
\end{tabular}
\end{center}
For the system  ${\mathcal H}_{\bf\Sigma}(@,\forall)$, under the  additional assumption that $\Gamma_s$ is $@$-witnessed,  we define 
the assignment $g^\Gamma : {\rm SVAR}\to W^\Gamma$ by  
 \vspace*{-0.2cm}
\begin{center}
\begin{tabular}{rcl}
$g_t^\Gamma(x)$ & $=$ & $|j|$ where $t\in S$, $x\in {\rm SVAR}_t$ and  $j\in {\rm NOM}_t$ such that $@_j^sx\in \Gamma_s$.
\end{tabular}
\end{center}
\end{definition}

\begin{lemma}\label{lema:henkin}
The Henkin model from Definition \ref{def:model} is well-defined.
\end{lemma}

\begin{proof}
Let $s\in S$ and assume that $\Gamma_s$ is a set of formulas of sort $s$.  Note that $R_\sigma^\Gamma$ is well-defined by $(Nom)$ and $(Bridge)$ from Lemma \ref{lemma:bridge}. For $t\in S$ and $j\in {\rm NOM}_t$, $V^\Gamma(j)$ is well-defined by axiom $(Ref)$. For the system  ${\mathcal H}_{\bf\Sigma}(@,\forall)$, we further   that $\Gamma_s$ is also  $@$-witnessed so, for any $t\in S$ and $x\in {\rm SVAR}_t$, there is a nominal $j\in {\rm NOM}_t $ such that $@_j^sx\in\Gamma$. The fact that $g^\Gamma$ is well-defined follows by $(Nom\, x)$.
\end{proof}

\begin{lemma}[Truth Lemma]\label{lemma:truth}
 \vspace*{-0.2cm}
\begin{itemize}
\item[1.] Let $s\in S$ and assume  $\Gamma_s$ is a named and pasted maximal consistent set of formulas of sort $s$ from ${\mathcal H}_{\bf\Sigma}(@)$. For any sort $t\in S$, $j\in {\rm NOM}_t\cup{\nom}_t$ and for any formula $\phi$ of sort $t$ we have ${\mathcal M}^\Gamma,|j|\mos{t}\phi \mbox { iff } @_j^s\phi\in\Gamma_s. $
\item[2.] Let $s\in S$ and assume  $\Gamma_s$ is a named, pasted and $@$-witnessed maximal consistent set of formulas of sort $s$ from ${\mathcal H}_{\bf\Sigma}(@,\forall)$. For any sort $s' \in S$, $j\in {\rm NOM}_{s'}\cup{\nom}_{s'}$ and for any formula $\phi$
of sort $s'$ we have ${\mathcal M}^\Gamma,g^\Gamma, |j|\mos{s'}\phi \mbox { iff } @_j^s\phi\in\Gamma_s$.
\end{itemize}
\end{lemma}

\begin{proof}
We make the proof by structural induction on $\phi$.
\begin{itemize}
\item ${\mathcal M}^\Gamma , |j| \mos{s'} a $, where $a\in {\rm PROP}_{s'}\cup {\rm NOM}_{s'}\cup {\nom}_{s'}$ iff $|j| \in V_{s'}^{\nom}(a) $ iff $ @_j^s a \in \Gamma_s $.

\item ${\mathcal M}^\Gamma , |j| \mos{s'} x$, where $x \in {\rm SVAR}_{s'}$ iff $g^{\Gamma}_{s'}(x)=|j|$ iff $@_j^s x \in \Gamma_s.$ 

\item $ {\mathcal M}^\Gamma , |j| \mos{s'} \neg \phi $ iff ${\mathcal M}^\Gamma , |j| \  \mosn{s'}  \phi $ iff $ @_j^s \phi \not \in \Gamma_s$, but we work with consistent sets, therefore  $ @_j^s \phi \not \in \Gamma_s $ iff $\neg  @_j^s  \phi\in \Gamma_s $ iff $ @_j^s \neg \phi  \in \Gamma_s$ $(SelfDual)$.

\item ${\mathcal M}^\Gamma , |j| \mos{s'} \phi \vee \varphi   $ iff ${\mathcal M}^\Gamma , |j| \mos{s'} \phi $ or ${\mathcal M}^\Gamma , |j| \mos{s'} \varphi  $ iff (inductive hypothesis) $ @_j^s \phi \in \Gamma_s$ or $ @_j^s\varphi   \in \Gamma_s $ iff $ @_j^s  \phi \vee  @_j^s  \varphi   \in\Gamma_s $ iff $ @_j^s  (\phi \vee \varphi  ) \in \Gamma_s  $.

\item ${\mathcal M}^\Gamma ,|j| \mos{s'} \sigma (\phi_1, \ldots, \phi_n)$  iff  exists $|k_i| \in W_{s_i}$ such that $R|j||k_1| \ldots |k_n|$ and ${\mathcal M}^\Gamma , |k_i| \mos{s_i} \phi_i$ for any $i \in [n]$. Using the induction hypothesis, we get $@_{k_i}^s \phi_i \in \Gamma_s$. But $R|j||k_1| \ldots |k_n|$ iff  $ @_j^s  \sigma(k_{1}, \ldots ,k_{n}) \in \Gamma_s$. Use the Bridge axiom to prove $  @_j^s  \sigma(k_{1}, \ldots ,k_{n}) \wedge @_{k_1}^s\phi_1  \wedge \ldots \wedge @_{k_n}^s\phi_n  \to  @_j^s  \sigma(\phi_{1}, \ldots ,\phi_{n})$, so  $ @_j^s  \sigma(\phi_{1}, \ldots ,\phi_{n}) \in \Gamma_s$. Now, suppose $ @_j^s  \sigma (\phi_1, \ldots, \phi_n) \in \Gamma_s$. We work with pasted models, so there are some nominals $k_i$ such that $ @_j^s  \sigma ( k_1, \ldots,  k_n)\in \Gamma_s$ and $ @_{k_i}^s \phi_i \in \Gamma_s$ for any $i \in [n]$. Therefore, exists $k_i$ such that $R|j||k_1|\ldots|k_n|$ and, by induction hypothesis,  ${\mathcal M}^\Gamma , |k_i| \mos{s_i} \phi_i$ for any $i \in [n]$ if and only if ${\mathcal M}^\Gamma ,|j| \mos{s'} \sigma (\phi_1, \ldots, \phi_n)$.

\item ${\mathcal M}^\Gamma , |j| \mos{s'} @_k^{s'} \phi $ iff $ {\mathcal M}^\Gamma , |k| \mos{s''} \phi$ , but from induction hypothesis $@_k^{s} \phi \in \Gamma_{s}$  and by applying $(Agree)$ we get $ @_j^s @_k^{s} \phi \in \Gamma_s$.
\end{itemize}

Further, for the ${\mathcal H}_{\bf\Sigma}(@,\forall)$ system, we need to pay attention to the assignment function and it only affects the following cases.

\begin{itemize}
\item  $ @_j^s \exists x \phi \in \Gamma_s$, then there exists $l \in {\rm NOM_{s'}}$ such that $  @_j^s  \phi[l/x] \in \Gamma_s$. Let $g'\stackrel{x}{\sim} g^\Gamma$  such that $g'_{s'}(x)=\{|l|\}$. Therefore, there exists $l \in {\rm NOM_{s'}}$ such that $g'_{s'}(x)=\{|l|\}$, $g'\stackrel{x}{\sim} g^\Gamma$ and ${\mathcal M}^\Gamma , g', |j| \mos{s'} \phi$ iff ${\mathcal M}^\Gamma ,g ^\Gamma,|j| \mos{s'}\exists x \phi$.
\item ${\mathcal M}^\Gamma,g^\Gamma, |j| \mos{s'} \exists x \phi$ iff exists  $g'\stackrel{x}{\sim} g^\Gamma$ and  ${\mathcal M}^\Gamma,g', |j| \mos{s'} \phi$. Let $g'_{s'}(x)=\{|l|\}$. Hence, there exists $l \in {\rm NOM_{s'}}$ such that $g'_{s'}(x)=\{|l|\}$, $g'\stackrel{x}{\sim} g^\Gamma$ and ${\mathcal M}^\Gamma , g', |j| \mos{s'} \phi$ iff ${\mathcal M}^\Gamma ,g,|j| \mos{s'}\phi[l/x]$ and from inductive hypothesis $ @_j^s  \phi[l/x] \in \Gamma_s$. Use the contrapositive of the $(Q2)$ axiom, \mbox{$\vds{s'} \phi[l/x] \to \exists x \phi$} and the $(Gen@)$ and $(K@)$ rules to obtain $ @_j^s  \phi[l/x] \to   @_j^s \exists x \phi \in \Gamma_s$. Therefore, $ @_j^s \exists x \phi \in \Gamma_s$.
\end{itemize}
\end{proof}

We are ready now to prove the strong completeness theorem for the hybrid logics ${\mathcal H}_{\bf\Sigma}(@)$ and ${\mathcal H}_{\bf\Sigma}(@,\forall)$ extended with pure axioms from $\Lambda$. For a logic $\mathcal L$, the relation $\vds{s}_{\mathcal L}$ denotes the local deduction,  the relation $\mos{s}_{Mod({\mathcal L})}$ denotes the semantic entailment w.r.t. models satisfying all the axioms of $\mathcal L$, while $\mos{s}_{\mathcal L}$ denotes the semantic entailment w.r.t. frames satisfying all the axioms of $\mathcal L$.

\begin{theorem}[Completeness]\label{th:comp} 
\begin{itemize}
\item[1.] {\sf Strong completeness}. 
 Let $s\in S$ and assume $\Gamma_s$ is a set of formulas of sort $s$. If $\Gamma_s$ is a consistent set in ${\mathcal L}={\mathcal H}_{\bf\Sigma}(@)$ {\rm (}in ${\mathcal L}={\mathcal H}_{\bf\Sigma}(@,\forall)${\rm )} then $\Gamma_s$ is satisfiable in a named  model. Consequently, for a formula $\phi$ of sort $s$,
$\Gamma_s {\mos{s}_{Mod(\mathcal L)}}\phi \mbox{ iff } \Gamma_s\vds{s}_{\mathcal L}\phi.$

\item[2.] {\sf Strong frame-completeness for pure extensions}. Let $\Lambda$ be a set of pure formulas in the language of ${\mathcal H}_{\bf\Sigma}(@)$ {\rm (}a set of $\forall$-pure formulas in the language of ${\mathcal H}_{\bf\Sigma}(@,\forall)${\rm )} and $s\in S$ and assume $\Gamma_s$ is a set of formulas of sort $s$. If $\Gamma_s$ is a consistent set in ${\mathcal L}={\mathcal H}_{\bf\Sigma}(@)+\Lambda$ {\rm (}in ${\mathcal L}={\mathcal H}_{\bf\Sigma}(@,\forall)+\Lambda${\rm )} then $\Gamma_s$ is satisfiable in a model  based on a frame that validates every formula in $\Lambda$. Consequently, for a formula $\phi$ of sort $s$,
$\Gamma_s {\mos{s}_{Mod(\mathcal L)}}\phi \mbox{ iff } \Gamma_s\vds{s}_{\mathcal L}\phi.$
\end{itemize}
\end{theorem}

\begin{proof}
Since 1. is obvious, we only prove 2.  If $\Gamma_s$ is a consistent set in ${\mathcal H}_{\bf\Sigma}(@,\forall)+\Lambda$ then, applying the Extended Lindenbaum Lemma, then $\Gamma_s\subseteq \Theta_s$, where $\Theta_s$ is a maximal consistent named, pasted and $@$-witnessed set (in an extended language ${\mathcal L}'$). If ${\mathcal M}^\Theta$ is the Henkin model and $g^\Theta$ is the assignment from Definition \ref{def:model} then, by Truth Lemma and (Intro)
 ${\mathcal M}^\Theta,g^\Theta, |j|\mos{s}\Gamma_s$ for any  $j\in {\rm NOM}_s\cup{\nom}_s$ such that $j\in \Gamma_s$. Moreover, ${\mathcal M}^\Theta$ is a named model (in the extended language) that is also a model of $\Lambda$.  By  Proposition \ref{prop:pure2}, the underlying frame of ${\mathcal M}^\Theta$ satisfies the $\forall$-pure formulas from $\Lambda$.  Hence the logic ${\mathcal H}_{\bf\Sigma}(@,\forall)+\Lambda$ is strongly complete w.r.t to the class of frames satisfying $\Lambda$.  Assume that $\Gamma_s {\mos{s}_{\Lambda}}\phi$ and suppose that $\Gamma_s\  \not\!\!\!\!\vds{s}\phi$. It follows that 
 $\Gamma_s\cup\{\neg\phi\}$ is consistent, so  $\Gamma_s$ is satisfied in a model based on a frame satisfying $\Lambda$ that is not a model of $\phi$. We get a contradiction, so the intended completeness result is proved. 
\end{proof}

The following useful results  can be easily proved semantically:

\begin{proposition}\label{lemma:nomConj}
\begin{itemize}
\item[1.] (Nominal Conjunction) For any formulas and any nominals of appropriate sorts, the following hold:
    
    \begin{itemize}
        \item[(i1)]$\sigma(\ldots, \phi_{i-1},\phi_i,\phi_{i+1}, \ldots) \wedge @_k(\psi) \lra \sigma(\ldots, \phi_{i-1},\phi_i \wedge @_k(\psi),\phi_{i+1}, \ldots)$
        \item[(i2)]$\sigma^{\mb}(\ldots, \phi_{i-1},\phi_i,\phi_{i+1}, \ldots) \wedge @_k(\psi) \lra$
        
        \hfill $ \sigma^{\mb}(\ldots, \phi_{i-1},\phi_i \wedge @_k(\psi),\phi_{i+1}, \ldots) \wedge @_k(\psi)$
    \end{itemize}
\item[2.] If $\phi_1, \ldots \phi_n$ are formulas of appropriate sorts and $x$ is a state variable that does not occur in $\phi_j$ for any $j \neq i$ then:
\begin{itemize}
\item[(i3)] $\exists x \sigma^{\mb} (\ldots, \phi_{i-1},\phi_i,\phi_{i+1}, \ldots) \to \sigma^{\mb} (,\ldots, \phi_{i-1},\exists x \phi_i,\phi_{i+1}, \ldots) $
\end{itemize}
\end{itemize}
\end{proposition}

\begin{proof}
1. (Nominal Conjunction) 
    \begin{itemize}
        \item[(i1)] ${\cal M}, g, w \mos{s} \sigma(\ldots, \phi_{i-1},\phi_i,\phi_{i+1}, \ldots) \wedge @_k(\psi)$ iff

        ${\cal M}, g, w \mos{s} @_k(\psi)$ and ${\cal M}, g, w \mos{s} \sigma(\ldots, \phi_{i-1},\phi_i,\phi_{i+1}, \ldots)$ iff

            ${\cal M}, g, v \mos{s'}  \psi$ where $V^{\nom}_{s'}=\{v\}$ and there exist $w_1\in W_{s_1},\ldots,w_n\in W_{s_n}$ such that $R_\sigma ww_1\cdots w_n$ and ${\cal M}, g, w_j \mos{s_j}\phi_j$ for all $1\leq j \leq n$ iff
            
            there exist $w_1\in W_{s_1},\ldots,w_n\in W_{s_n}$ such that $R_\sigma ww_1\cdots w_n$ and \\${\cal M}, g, w_j \mos{s_j} \phi_j$ for all $1\leq j \leq n$, $j\neq i$, and ${\cal M}, g, w_i \mos{s_i} \phi_i\wedge @_k(\psi)$ iff

        ${\cal M}, g, w \models \sigma(\ldots, \phi_{i-1},\phi_i \wedge @_k(\psi),\phi_{i+1}, \ldots)$

        \item[(i2)] ${\cal M}, g, w \mos{s} \sigma^{\mb}(\ldots, \phi_{i-1},\phi_i,\phi_{i+1}, \ldots) \wedge @_k(\psi)$ iff

        ${\cal M}, g, w \mos{s} @_k(\psi)$ and ${\cal M}, g, w \mos{s} \neg\sigma(\ldots, \neg\phi_{i-1},\neg\phi_i,\neg\phi_{i+1}, \ldots)$ iff

            ${\cal M}, g, v \mos{s'}  \psi$ where $V^{\nom}_{s'}=\{v\}$ and for all $w_1\in W_{s_1},\ldots,w_n\in W_{s_n}$ for which $R_\sigma ww_1\cdots w_n$, there exists $1\leq j \leq n$ such that ${\cal M}, g, w_j \mos{s_j} \phi_j$ iff 

            ${\cal M}, g, v \mos{s'}\psi$ where $V^{\nom}_{s'}=\{v\}$ and for all $w_1\in W_{s_1},\ldots,w_n\in W_{s_n}$ for which $R_\sigma ww_1\cdots w_n$, there exists $1\leq j \leq n$, $j\neq i$ such that ${\cal M}, g, w_j \mos{s_j} \phi_j$ or ${\cal M}, g, w_i \mos{s_i} \phi_i$ iff 
${\cal M}, g, v \mos{s'} \psi$ and for all $w_1\in W_{s_1},\ldots,w_n\in W_{s_n}$ for which $R_\sigma ww_1\cdots w_n$, there exists $1\leq j \leq n$, $j\neq i$ such that ${\cal M}, g, w_j \mos{s_j} \phi_j$ or ${\cal M}, g, w_i \mos{s_i}\phi_i \wedge @_k(\psi)$ iff 
${\cal M}, g, w \mos{s} \sigma^{\mb}(\ldots, \phi_{i-1},\phi_i \wedge @_k(\psi),\phi_{i+1}, \ldots) \wedge @_k(\psi)$
\end{itemize}
2.
\begin{itemize}
\item[(i3)] $\mathcal{M}, g, w \mos{s} \exists x \sigma^{\mb} (\phi_1,\ldots, \phi_{i-1},\phi_i,\phi_{i+1}, \ldots, \phi_n)$ iff exists $g' \stackrel{x}{\sim} g$ such that $\mathcal{M}, g', w \mos{s}  \sigma^{\mb} (\phi_1,\ldots, \phi_{i-1},\phi_i,\phi_{i+1}, \ldots, \phi_n)$ iff exists $g' \stackrel{x}{\sim} g$ such that for all $v_j \in W_{s_j}$, $R_{\sigma}wv_1\ldots v_n$ implies $\mathcal{M}, g', v_j \mos{s_j}  \phi_j$ for some $j \in [n]$. Then, for all $v_j \in W_{s_j}$, $R_{\sigma}wv_1\ldots v_n$ implies there exists $g' \stackrel{x}{\sim} g$ such that $\mathcal{M}, g', v_j \mos{s_j}  \phi_j$ for some $j \in [n]$. But $x$ does not occur in $\phi_j$ for any $j \in [n]$ and $j \neq i$, so for all $v_j \in W_{s_j}$ and $v_i \in W_{s_i}$, $R_{\sigma}wv_1\ldots v_i \ldots v_n$ implies $\mathcal{M}, g', v_j \mos{s_j}  \phi_j$ and there exists $g' \stackrel{x}{\sim} g$ such that $\mathcal{M}, g', v_i \mos{s_i}  \phi_i$ for some $i,j \in [n]$ and $j \neq i$. We use Agreement Lemma, for all $v_j \in W_{s_j}$ and $v_i \in W_{s_i}$, $R_{\sigma}wv_1\ldots v_i \ldots v_n$ implies $\mathcal{M}, g, v_j \mos{s_j}  \phi_j$ and $\mathcal{M}, g, v_i \mos{s_i} \exists x \phi_i$ for some $i,j \in [n]$ and $j \neq i$. Therefore, $\mathcal{M}, g,w \mos{s} \sigma^{\mb} (\phi_1,\ldots, \phi_{i-1},\exists x \phi_i,\phi_{i+1}, \ldots, \phi_n) $. 
\end{itemize}
\end{proof}

In the many-sorted setting one can wonder what happens if we have an $S$-sorted set of deduction hypothesis ${\mathbf\Gamma}=\{\Gamma_s\}_{s\in S}$.  The following considerations hold for any of   ${\mathcal H}_{\bf\Sigma}(@)$ and  ${\mathcal H}_{\bf\Sigma}(@,\forall)$.  Clearly, a model $\mathcal M$ is a model of $\mathbf \Gamma$ if ${\mathcal M}\mos{s}\gamma_s$ for any $s\in S$ and $\gamma_s\in \Gamma_s$ (in this case we write ${\mathcal M}\models{\mathbf \Gamma}$). Using the "broadcasting" properties of the $@_i$ operators, we define another syntactic consequence relation:

\begin{center}
{\em ${\mathbf \Gamma}\vss{s} \phi$ iff there are $s_1,\ldots, s_n\in S$, $j_1\in {\rm NOM}_{s_1},\ldots, j_n\in {\rm NOM}_{s_n}$ and $\gamma_1\in \Gamma_{s_1},\ldots, \gamma_n\in \Gamma_{s_n}$ such that $\vds{s}@_{j_1}^s\gamma_1\wedge\cdots \wedge @_{j_n}^s\gamma_n\to \phi$}. 
\end{center}

\begin{proposition}[$\vss{s}$ soundness]\label{prop:newded}
Let $\mathbf \Gamma$ be an $S$-sorted set and $\phi$  a formula of sort $s\in S$. If  ${\mathbf \Gamma}\vss{s} \varphi$  then ${\mathcal M}\models{\mathbf \Gamma}$ implies ${\mathcal M}\mos{s}\phi$ for any model $\mathcal M$. 
\end{proposition}

\begin{proof}
 Let $\mathcal M$ be a model and assume 
$\vds{s}@_{j_1}^s\gamma_1\wedge\cdots \wedge
 @_{j_n}^s\gamma_n\to \phi$ as above. If  ${\mathcal M}\models{\mathbf \Gamma}$ then, by ($Gen@$), ${\mathcal M}\mos{s}{\Gamma_s}\cup\{@_{j_1}^s\gamma_1,\ldots ,
 @_{j_n}^s\gamma_n\}$. Using the soundness of the local deduction, we get the desired conclusion. 
\end{proof}

\section{A SMC-like language and a Hoare-like logic for it}\label{sec:app}

To showcase the application of our logic into program verification, we have chosen to specify 
a state-machine, whose expressions have side effects and where Hoare-like semantics are known to
be hard to use. 

In Figure \ref{fig:lang}, we introduce the signature ${\bf\Sigma}=(S,\Sigma,\nom)$ of our logic as a context-free grammar (CFG) in a BNF-like form.
We make use of the established equivalence between CFGs and algebraic signatures (see, e.g., \cite{HHKR89}), by mapping non-terminals to sorts and CFG productions to operation symbols.
Note that, due to non-terminal renamings (e.g., \texttt{AExp ::= Nat}), it may seem that our syntax relies on subsorting. However, this is done for readability reasons only. The renaming of non-terminals in syntax can be thought of as syntactic sugar for defining injection functions. For example, \texttt{AExp ::= Nat} can be thought of as \texttt{AExp ::= nat2Exp(Nat)}, and all occurrences of an integer term in a context, in which an expression is expected, could be wrapped by the \texttt{nat2Exp} function.

Our language is inspired by the \textit{SMC machine}~\cite{plotkin} which consists of a set of transition rules defined between configurations of the form $\left\langle S, M, C \right\rangle$, where $S$ is the {\em value stack} of intermediate results, $M$ represents the {\em memory}, mapping program identifiers to values, and $C$ is a {\em control stack} of commands representing the control flow of the program.
Since our target is to extend the Propositional Dynamic Logic (PDL) \cite{dynamic}, we identify the control stack with the notion of {\em program} in dynamic logic, and use the ";" operator to denote stack composition. We define our formulas to stand for {\em configurations} of the form $\langle vs, mem\rangle$ comprising only of a value stack and a memory. Hence, the sorts ${CtrlStack}$ and ${Config}$ correspond to programs and formulas from PDL, respectively. 
 Inspired by  PDL, we use the dual modal operator
$ [\_]\_ : CtrlStack \times Config \to Config$
to assert that a configuration formula must hold after executing the commands in the control stack.
The axioms defining the dynamic logic semantics of the SMC machine are then formulas of the form 
$cfg \to [ctrl] cfg'$ 
saying that a configuration satisfying $cfg$ must change to one satisfying $cfg'$ after executing $ctrl$.
The usual operations of dynamic logic $;$, $\cup$, $^*$ are defined accordingly \cite[Chapter 5]{dynamic}. We depart from PDL with the definition of $?$ (test): in our setting, in order to take a decision, we test the top value of the value stack. Consequently,   the signature of the test operator is $?:Val\to CtrlStack$.

A deductive system, that allows us to accomplish our goal, is defined in Figure \ref{fig:lang}. In this way we define  an expansion of ${\cal H}(@,\forall)$.
Our definition is incomplete (e.g. we do not fully axiomatize the natural numbers), but one can see that,e.g. $\nom_{Bool}=\{\Strue,\Sfalse\}$. 
To simplify the presentation, we omit sort annotations in the sequel; these should be easily inferrable from the context.

\begin{figure}
\small{
{\bf Domains}

\vspace*{-0.15cm}

\begin{verbatim}
 Nat ::=  natural numbers
 Bool ::= true | false | Nat == Nat | Nat <= Nat
\end{verbatim}

\noindent \begin{minipage}[t]{.5\textwidth}
{\bf Syntax}

\vspace*{-0.15cm}

\begin{verbatim}
Var ::=  program variables
AExp ::= Nat | Var | AExp + AExp
        | ++ Var
BExp ::= AExp <= AExp
Stmt ::= x := AExp
        | if BExp 
          then Stmt 
          else Stmt
        | while BExp do Stmt
        | skip
        | Stmt ; Stmt
\end{verbatim}
\end{minipage}
\begin{minipage}[t]{.5\textwidth}
{\bf Semantics}

\vspace*{-0.15cm}

\begin{verbatim}
     Val ::= Nat | Bool
ValStack ::= nil 
             | Val . ValStack
     Mem ::= empty | set(Mem, x, n)    
CtrlStack ::= c(AExp)
             | c(BExp)  
             | c(Stmt)
             | asgn(x)   
             | plus   | leq    
             | Val ?
             | c1 ; c2
  Config ::=  < ValStack, Mem >

\end{verbatim}
\end{minipage}
}

\noindent{\bf Domains axioms (incomplete)}

\vspace*{0.15cm}

$\begin{array}{llll}
(B1) & \Strue \lra \neg\Sfalse & \hspace*{1cm}(I1) & @_{\Strue}^{Nat}(x == y) \to (x\lra y)\\
& \ldots & & \ldots
\end{array}$

\medskip
\noindent{\bf PDL-inspired axioms }

\vspace*{0.15cm}

$\begin{array}{llll}
(A\cup) & [\pi \cup \pi'] \gamma \leftrightarrow [\pi] \gamma \wedge [\pi'] \gamma  &
(A;) & [\pi;\pi'] \gamma \leftrightarrow [\pi][ \pi'] \gamma \\
(A?) &  \cf{v \cdot  vs, mem} \to [v ?] \cf{vs,mem}  &
(A\neg ?) &\cf{v \cdot  vs, mem} \wedge @_v(\neg v') \to [v' ?] \bot\\
(A^*) & [\pi^*] \gamma \leftrightarrow \gamma \wedge [\pi][\pi^*] \gamma & 
(A{\it Ind}) & \gamma \wedge [\pi^*](\gamma \to [\pi]\gamma) \to [\pi^*] \gamma
\end{array}$

\noindent Here,  $\pi$, $\pi'$ are formulas of sort 
${CtrlStack}$ ("programs"), $\gamma$ is a  formula of sort $Config$ (the analogue of "formulas" from PDL), $v$ and $v'$ are state variables of sort $Var$, $vs$ has the sort 
$ValStack$ and $mem$ has the sort $Mem$.

\medskip

\noindent{\bf SMC-inspired axioms}

\vspace*{0.15cm}
 
$\begin{array}{ll}
(CStmt)& c(s1;s2)\lra c(s1);c(s2)\\ 
(Aint) &   \cf{vs,mem} \to [c(n)] \cf{n \cdot vs,mem}    \mbox{ where } n \mbox{ is an integer}\\
(Aid)  &  \cf{vs, set(mem,x,n}) \to [c(x)] \cf{n \cdot vs,set(mem,x,n})\\
(A++)  &  \cf{vs, set(mem,x,n}) \to [c(++x)] \cf{n+1 \cdot vs,set(mem,x,n+1})\\
(Dplus) &    c(a1 + a2) \lra c(a1) ; c(a2) ; \Splus \\
(Aplus)  &  \cf{n2 \cdot n1 \cdot vs,mem} \to [\Splus] \cf{(n1+n2} \cdot vs,mem)\\
\end{array} $

$\begin{array}{ll}
(Dleq)   & c(a1 <= a2) \lra c(a2) ; c(a1) ; \Sleq \\
(Aleq)    &\cf{n1 \cdot n2 \cdot vs,mem} \to  [\Sleq] \cf{(n1\leq n2} \cdot vs,mem) \\

(Askip) &    \gamma \to  [c(\Sskip)]\gamma \\

(Dasgn) &   c(x := a) \lra c(a) ; \Sasgn(x)\\
(Aasgn)  &  \cf{n \cdot vs,mem} \to [\Sasgn(x)] \cf{vs,set(mem, x, n})\\

(Dif)  &  c(\Sif b \Sthen s1 \Selse s2) \lra c(b) ;
       ( (\Strue \, ?; c(s1))\cup (\Sfalse \, ?; c(s2)) ) \\

(Dwhile)  & c(\Swhile\,\, b\,\, do\,\, s) \lra c(b) ;(\Strue ? ; c(s); c(b))^*; \Sfalse ?
\end{array}$

\medskip

\noindent{\bf Memory consistency axioms}\nopagebreak 
 
 \vspace*{0.15cm}
 
$\begin{array}{ll}
(AMem1) & set(set(mem, x, n),y,m)\lra set(set(mem, y,m),x,n) \\
&\mbox{where } x \mbox{ and } y \mbox{ are distinct}\\
(AMem2) & set(set(mem, x, n),x,m)\to set(mem, x,m) \\

\end{array}$
\caption{\bf Axioms defining an SMC-like programming language}\label{fig:lang}
\end{figure}

\begin{remark}
Assume that  $\Lambda$ contains all the axioms from Figure  \label{fig:lang}
 and denote ${\bf{\cal L}}={\cal H}(@,\forall)+\Lambda$.
Then ${\bf{\cal L}}$ is a many-valued hybrid modal system associated to our language, and all results from Section \ref{sec:hybridgen} applies in this case.
\end{remark}


 We present below several  {\em Hoare-like rules of inference}.  Note that they are provable from the PDL and language axioms.

\begin{proposition}\label{lem:adhoare} The following rules are admissible :
\begin{itemize}
\item[1.] {\bf Rules of Consequence} \label{lemma:weak}
\begin{itemize}
    \item[]   If $\vdash \phi \to[\alpha] \psi$ and $\vdash\psi\to \chi$ then $\vdash\phi \to[\alpha] \chi$.

    \item[]     If $\vdash\phi \to[\alpha] \psi$ and $\vdash\chi\to \phi$ then $\vdash \chi \to[\alpha] \psi$.
\end{itemize}

\item[2.] {\bf Rule of Composition, iterated}\label{lemma:seq}
    \begin{itemize}
        \item[] If $\phi_0 \to[\alpha_1]\phi_1$, \ldots, $\phi_{n-1} \to[\alpha_n]\phi_n$, then $\phi_0 \to [\alpha_1 ; \ldots ; \alpha_n] \phi_n$.
    \end{itemize}

\item[3.] {\bf Rule of Conditional} \label{lemma:if}

 If  $B$ is a formula of sort $Bool$, and  $vs$, $mem$, $P$ are formulas of appropriate sorts such that 
 
 \begin{tabular}{ll}
 (h1) &        $\vdash \phi \to[c(b)] (\cf{B \cdot vs,mem}\wedge P) $,\\
  
 (h2) &        $\vdash \cf{vs,mem}\wedge P \wedge @_{\Strue}(B) \to [c(s1)]\chi$ \\
              
 (h3) &           $\vdash \cf{vs,mem}\wedge P \wedge @_{\Sfalse}(B) \to [c(s2)]\chi$\\
    
(h4) &  $\vdash  P   \to [\alpha]P$ for any $\alpha$ of sort $CtrlStack$,
    \end{tabular}
    
then   $\vdash \phi \to [c(\Sif b \Sthen s1 \Selse s2)]\chi$
\end{itemize}
\end{proposition}

\begin{proof}
In the sequel we shall mention the sort of a formula only when it is necessary. 
\begin{itemize}
\item[1.] Rule of Consequence follows easily by $(UG)$.
\item[2.] Rule of Composition follows easily by $(UG)$ and $(CStmt)$.
\item[3.] Rule of Conditional.  Since $B$ is a formula of sort $Bool$, using the axiom $(B1)$ and the completeness theorem,  one can easily infer that 

$\vdash B\lra (\Strue\wedge @_{\Strue}B)\vee (\Sfalse\wedge @_{\Sfalse}B) $ \\
Using the fact that any operator $\sigma\in \Sigma$ commutes with disjunctions, Lemma \ref{lemma:nomConj} we get 

($\ast$) $\vdash \cf{B \cdot vs,mem}\to(\cf{\Strue\cdot vs,mem} \wedge @_{\Strue}B)\vee$

\hspace*{3cm} $(\cf{\Sfalse\cdot vs,mem} \wedge @_{\Sfalse}B)$ 

\medskip

Now we prove that 

$\vdash \cf{\Strue \cdot vs,mem}\wedge @_{\Strue}B\to[(\Strue?;c(s1))\cup (\Sfalse;c(s2))]\chi$.

Note that $\vdash @_{\Strue}(\neg\Sfalse)$, so we use $(A?)$ and $(A\neg?)$ as follows: 

$\vdash \cf{\Strue\cdot vs,mem} \wedge @_{\Strue}B\to 
\cf{\Strue\cdot vs,mem} \wedge @_{\Strue}B\wedge @_{\Strue}(\neg\Sfalse)$

$\vdash \cf{\Strue\cdot vs,mem} \to [\Strue ?]\cf{vs,mem}$
 
$\vdash \cf{\Strue\cdot vs,mem} \wedge @_{\Strue}(\neg\Sfalse)\to [\Sfalse ?]\bot$

Next we prove that 

($@[]$) $\vdash @_k\varphi\to [\alpha]@_k\varphi$

for any  formulas $\alpha$, $\varphi$ and  nominal $k$ of appropriate sorts.  Note that $\vdash [\alpha]\top$ so,  using  Lemma \ref{lemma:nomConj}, we have the following chain of inferences:

$\vdash @_k\varphi\to @_k\varphi\wedge [\alpha]\top$

$\vdash @_k\varphi\wedge [\alpha]\top\to [\alpha]@_k\varphi$

and $(@[])$ easily follows.

Consequently 
$\vdash @_{\Strue}B\to [\Strue?]@_{\Strue}B$

Since dual operators $\sigma^{\mb}$ for $\sigma\in \Sigma$ commutes with conjunctions, using also  (h4)  we get 

$\vdash \cf{\Strue \cdot vs,mem}\wedge P\wedge @_{\Strue}B\to([\Strue?](\cf{vs,mem} \wedge P\wedge @_{\Strue}B)) \wedge [\Sfalse ?]\bot$

By (h2) and $(K)$ it follows that 

 $\vdash \cf{\Strue \cdot vs,mem}\wedge \wedge P @_{\Strue}B\to[\Strue?;c(s1)]\chi  \wedge [\Sfalse ?]\bot$ 
 
Since $\bot\to [c(s2)]\chi$, and using $(A\cup)$ we proved     
     
$\vdash \cf{\Strue \cdot vs,mem}\wedge P\wedge @_{\Strue}B\to[(\Strue?;c(s1))\cup (\Sfalse ?;c(s2))]\chi$.

In a similar way, we get 

$\vdash \cf{\Sfalse \cdot vs,mem}\wedge P\wedge @_{\Sfalse}B\to[(\Strue?;c(s1))\cup (\Sfalse ?;c(s2))]\chi$.

By ($\ast$)  we infer 

$\vdash \cf{B \cdot vs,mem}\to [(\Strue?;c(s1))\cup (\Sfalse ?;c(s2))]\chi$

Using $(K)$ and $(Dif)$ we get the conclusion. 
\end{itemize}
\vspace*{-0.4cm}
\end{proof}
\vspace*{-0.2cm}
Note that our Rule of Conditional requires two more hypotheses, (h1) and (h4) than the inspiring rule in Hoare-logic.
(h1) is needed because language expressions are no longer identical to formulas and need to be evaluated; in particular this allows for expressions to have side effects.
(h4) is useful to carry over extra conditions through the rule; note that (h4) holds for all $@_j\varphi$ formulas.

Similarly, the Rule of Iteration needs to take into account the evaluation steps required for
evaluating the condition. Moreover, since assignment is now handled by a forwards-going operational
rule, we require existential quantification over the invariant to account for the values of the
program variables in the memory, and work
with instances of the existentially quantified variables.

\begin{proposition}[Rule of Iteration] \label{lemma:while}
     Let
    $B$, $vs$, $mem$, and $P$ be formulas with variables over $\ora{x}$, where $\ora{x}$ is a set of state variables.
If there exist substitutions $\ora{x_{init}}$ and $\ora{x_{body}}$ for the variables of $\ora{x}$ such that:

(h1) $\vdash \phi \to [c(b)] (\cf{B\cdot vs, mem}\wedge P)[\ora{x_{init}} / \ora{x}]$, 

(h2) $\vdash \cf{vs, mem}\wedge P \wedge @_{\Strue}(B)\to [c(s); c(b)] (\cf{B\cdot vs, mem} \wedge P)[\ora{x_{body}} / \ora{x}]$

(h3) $\vdash P\to [\alpha] P$ for any formula $\alpha$ of sort $CtrlStack$

\noindent  then $\vdash \phi \to [c(\Swhile b \Sdo s)] \exists \ora{x} \cf{vs, mem}\wedge P \wedge @_{\Sfalse}(B)$. 
 \end{proposition}

\begin{proof}
Denote $\theta := \cf{B\cdot vs, mem}\wedge P$ and 
 $\theta_I := \exists \ora{x}\theta$. We think of $\theta_I$ as being the invariant of $\Swhile b \Sdo s$.
Note that, using the contraposition of $(Q2)$ and  (h1) we infer that 

(c1) $\vdash \phi\to [c(b)]\theta_I$

\noindent In the following we firstly prove that

(c2) $\vdash \theta_I\to [\alpha]\theta_I$,\\ 
where $\alpha = \Strue ?;c(s);c(b)$. Since 

$\vdash B\lra (\Strue\wedge @_{\Strue}B)\vee (\Sfalse\wedge @_{\Sfalse} B)$  

\noindent it follows that

 $\vdash \theta \to (\cf{\Strue \cdot vs,mem}\wedge P\wedge @_{\Strue} B)\vee  (\cf{\Sfalse\cdot vs,mem}\wedge P\wedge @_{\Sfalse} B)$
 
\noindent By $(A?)$, (h3) and ($@[]$)  (from the proof of Proposition \ref{lem:adhoare}) we infer
 
 $\vdash \cf{\Strue \cdot vs,mem}\wedge P\wedge @_{\Strue} B\to [true?](\cf{vs,mem}\wedge P\wedge @_{\Strue} B)$
 
 and, by (h2)
 
 $\vdash \cf{\Strue \cdot vs,mem}\wedge P\wedge @_{\Strue} B)\to [\alpha]\theta [\ora{x_{body}}/\ora{x}]$
 
 Since $\vdash @_{\Sfalse}(\neg\Strue)$, by  $(A\neg?)$ we get 

 $\vdash \cf{\Sfalse \cdot vs,mem}\wedge  @_{\Sfalse}(\neg\Strue)\to [\Strue?]\bot$, so 
 
  $\vdash \cf{\Sfalse \cdot vs,mem}\wedge P\wedge @_{\Sfalse} B)\to [\alpha]\theta [\ora{x_{body}}/\ora{x}]$
  
\noindent As consequence  $\vdash  \theta \to [\alpha]\theta [\ora{x_{body}}/\ora{x}]$ and, using the contraposition of $Q_2$,\\ we infer that $\theta\to [\alpha]\theta_I$. We use now the fact that 
 
$\vdash\forall x(\varphi(x)\to \psi)\to (\exists x \varphi(x)\to\psi)$ if $x$ does not appear in  $\psi$,

\noindent which leads us to $\vdash \theta_I\to[\alpha]\theta_I$. Using $(UG)$ we get $\vdash [c(b);\alpha^*] (\theta_I\to [\alpha]\theta_I)$. 

By (c1) it follows that  

$\vdash \phi\to ([c(b)]\theta_I \wedge ([c(b);\alpha^*] (\theta_I\to [\alpha]\theta_I)) $

Using the induction axiom, $(UG)$, $(K)$ and the fact that the dual operators commutes with conjunctions, we get 

$\vdash ([c(b)]\theta_I \wedge ([c(b);\alpha^*] (\theta_I\to [\alpha]\theta_I))\to [c(b);\alpha^*]\theta_I$

\noindent so $\vdash \phi\to [c(b);\alpha^*]\theta_I$, which proves the invariant property of $\Swhile b \Sdo s$.  

To conclude, so far we proved

$\vdash \phi\to [c(b);\alpha^*] \exists \ora{x}\theta$

We can safely assume that the state variables from $\ora{x}$ do not appear in $\phi$, $b$

Note that $c(\Swhile b \Sdo s)\lra c(b);\alpha^*; \Sfalse ?$
  
As before, 
  
$\vdash \theta \to (\cf{\Strue \cdot vs,mem}\wedge P\wedge @_{\Strue} B)\vee  (\cf{\Sfalse\cdot vs,mem}\wedge P\wedge @_{\Sfalse} B)$  

Using again $(A?)$ and $(A\neg?)$ we have that 

$\vdash \cf{\Sfalse\cdot vs,mem}\to [\Sfalse ?]\cf{vs,mem}$

$\vdash \cf{\Strue\cdot vs,mem}\wedge @_{\Strue}(\neg\Sfalse)\to [\Sfalse ?]\bot$

\noindent It follows that 

$\vdash \theta \to [\Sfalse ?](<vs, mem>\wedge P\wedge @_{\Sfalse}B)$ so, using the properties of the existential binder 


$\vdash \exists\ora{x}\theta \to \exists\ora{x}[\Sfalse ?](<vs, mem>\wedge P\wedge @_{\Sfalse}B)$

Since the state variables from $\ora{x}$ do not appear in $\Sfalse ?$,  by Lemma \ref{lemma:nomConj} it follows that

$\vdash \exists\ora{x}[\Sfalse ?](<vs, mem>\wedge P\wedge @_{\Sfalse}B)\to$

\hfill$ [\Sfalse ?] \exists\ora{x}(<vs, mem>\wedge P\wedge @_{\Sfalse}B)$

We can finally obtain the intended result:

$\vdash \phi\to [c(b);\alpha^*;\Sfalse?] \exists\ora{x}(<vs, mem>\wedge P\wedge @_{\Sfalse}B)$
\end{proof}

\paragraph{Proving a program correct. }
Let us now exhibit proving a program using the operational semantics and the Hoare-like rules above.  Consider the program:
\begin{lstlisting}[language=C]
s := 0; i := 0;
while ++ i <= n do s := s + i;
\end{lstlisting}

Let $\Spgm$ stand for the entire program. We want to prove that if the initial value of $n$ is any natural number, then the final value of $s$ is the sum of numbers from $1$ to $n$.
Formally,

$\cf{vs,\Sset(\Smem,n,vn}\to{}$

\hfill $[c(\Spgm)] \cf{vs, \Sset(\Sset(\Sset(\Smem,n,vn},s,vn*(vn+1)/2),i,vn+1))$

Let $\SB$ stand for $++i <= n$ and $\SC$ stand for $s := s + i$.
By applying the axioms above we can decompose $\Spgm$ as
\vspace*{-0.2cm}
\begin{center}
$c(pgm) \lra c(0) ; \Sasgn(s) ; c(0) ; \Sasgn(i) ; c(\Swhile \SB \Sdo \SC)$
\end{center}
\vspace*{-0.2cm}
Similarly, $c(\SB) \lra c(++i);c(n);\Sleq$ and $c(\SC) \lra c(s);c(i);\Splus;\Sasgn(s)$.

We have the following instantiations of the axioms:

\noindent $\cf{vs,\Sset(\Smem,n,vn)}\to [c(0)]\cf{0\cdot vs,\Sset(\Smem,n,vn)}$ 
 \hfill $Aint$

\noindent $\cf{0\cdot vs,\Sset(\Smem,n,vn)} \to [asgn(s)]\cf{vs,\Sset(\Sset(\Smem,n,vn), s,0)}$  
\hfill $Aasgn$

\noindent $\cf{vs,\Sset(\Sset(\Smem,n,vn),s,0)} \to [c(0)]\cf{0\cdot vs, \Sset(\Sset(\Smem,n,vn),s,0)}$
 \hfill $Aint$

\noindent $\cf{0\cdot vs, \Sset(\Sset(\Smem,n,vn}, s,0))$

\hfill ${}\to[asgn(i)]
\cf{vs, \Sset(\Sset(\Sset(\Smem,n,vn), s,0), i,0)}$
\hfill  $Aasgn$

And by applying the Rule of Composition we obtain:

\noindent(1) $\cf{vs,\Sset(\Smem,n,vn})$

\hfill $\to[c(0) ; \Sasgn(s) ; c(0) ; \Sasgn(i)]\cf{vs, \Sset(\Sset(\Sset(\Smem,n,vn), s,0), i,0)}$

We now want to apply the Rule of Iteration. First let us handle the condition.
Similarly to the ``stepping'' sequence above,
we can use instances of (A++), (Aid), (Aleq), and the Rule of Composition to chain them to obtain:

\noindent $\cf{vs, \Sset(\Sset(\Sset(\Smem,n,vn), s,0), i,0)}$

\hfill $\to [c(\SC)]\cf{(1\leq vn)\cdot vs, \Sset(\Sset(\Sset(\Smem, s,0), i,1),n,vn)}$

Let $\mathbf{x} = vi$,
$B = vi\leq vn$,
${\it vs} = vs$,
${\it mem} = \Sset(\Sset(\Sset(\Smem, s,(vi-1)*vi/2), i,vi),n,vn)$,
$P = @_{\Strue}(vi\leq vn+1)$.
For $\mathbf{x_{\bf init}} = 1$ we have that 
$B[1/vi] = 1\leq vn$,
${\it mem}[1/vi] = \Sset(\Sset(\Sset(\Smem, s,(1-1)*1/2), i,1),n,vn)$,
$P[1/vi] = @_{\Strue}(1\leq vn+1)$.
Using that $(1-1)*1/2 \lra 0$ and $1\leq vn+1$ we obtain

\noindent(2) $\cf{vs, \Sset(\Sset(\Sset(\Smem,n,vn), s,0), i,0)}\to [c(\SB)](\cf{B\cdot {\it vs}, {\it mem}}\wedge P)[1/vi]$


Now, we can again use instances of (Aid), (Aid), (Aplus), (Aasgn), (AMem), (A++), (AId), (Aleq),
and the Rule of Composition to derive

\noindent $\cf{vs, \Sset(\Sset(\Sset(\Smem, i,vi),n,vn), s,(vi-1)*vi/2)}\to[c(\SC);c(\SB)]$ 

\noindent \hfill $\cf{(vi+1 \leq vn)\cdot vs, \Sset(\Sset(\Sset(\Smem, s,vi*(vi+1)/2, i,vi+1),n,vn)}$

By applying equivalences between formulas on naturals, the above leads to

\noindent $\cf{vs, \Sset(\Sset(\Sset(\Smem, i,vi),n,vn), s,(vi-1)*vi/2)}$

\noindent \hfill $\to[c(\SC);c(\SB)]\cf{B\cdot vs, mem}[vi+1/vi]$

\noindent Using Proposition \ref{lemma:nomConj} $(i2)$ and the fact that $vi \leq vn\lra vi+1 \leq vn+1$, we obtain

\noindent(3)$\cf{B\cdot vs, mem} \wedge P \wedge @_{\Strue}(B)$

\noindent \hfill $\to[c(\SC);c(\SB)](\cf{B\cdot vs, mem} \wedge P)[vi+1/vi]$

\medskip

Now using the Rule of Iteration with (2) and (3) we derive that

\noindent $\cf{vs, \Sset(\Sset(\Sset(\Smem,n,vn), s,0), i,0)}$

\noindent \hfill $\to[c(\Swhile \SB \Sdo \SC)]\exists vi.\cf{B\cdot vs, mem} \wedge P \wedge @_{\Sfalse}(B)$

By arithmetic reasoning, $\vdash ({\Sfalse}\to vi\leq vn) \lra ({\Strue} \to vn+1\leq vi)$, hence $\vdash @_{\Sfalse}(vi\leq vn) \lra @_{\Strue}(vn+1\leq vi)$.
Moreover, $@_{\Strue}(vn+1\leq vi)\wedge @_{\Strue}(vi\leq vn+1) \lra @_{\Strue}(vn+1\leq vi\wedge vi\leq vn+1)$ which by arithmetic reasoning is equivalent to $@_{\Strue}(vi =_{\it Nat} vn+1)$, which by (I1) is equivalent to $vi \lra vn+1$ which allows us to substitute $vi$ by $vn+1$ and eliminate the quantification, leading to 
\vspace*{-0.2cm}
 \begin{center}
 $\exists vi.\cf{vs, mem} \wedge P \wedge @_{\Sfalse}(B) \lra \cf{vs, mem}[vn+1/vi],
\mbox{ hence,}$
 \end{center}
 \vspace*{-0.2cm}
\noindent(4) $\cf{vs, mem'}\to[c(\Swhile \SB \Sdo \SC)]\cf{ vs, mem''}$\\
where $mem'' = \Sset(\Sset(\Sset(\Smem, s,vn*(vn+1)/2),
 i,vn+1), n,vn)$,\\
\hspace*{1cm}$mem' = \Sset(\Sset(\Sset(\Smem,n,vn), s,0), i,0)$.

Using the Rule of Composition on  (1) and (4) we obtain our goal.

\section{Conclusions and related work}\label{sec:final}

We defined a general many-sorted hybrid polyadic modal logic that is sound and complete with respect to the usual modal semantics. From a theoretical point of view, we introduced nominal constants and we restricted the application of the satisfaction operators to nominals alone. We proved that the system is sound and complete and we also investigated the completeness of its pure axiomatic expansions. Given a concrete language with a concrete SMC-inspired operational semantics, we showed how to define a corresponding (sound and complete) logical system and we also proved (rather general) results that allow us to perform Hoare-style verification.  
Our approach was to define the weakest system that allows us to reach our goals. 

There is an abundance of research literature on hybrid modal logic, we refer to \cite{hand} for a comprehensive overview. Our work was mostly inspired by \cite{pureax,hybtemp,goranko,goranko2}, where a variety of hybrid modal logics are studied in a mono-sorted setting. We need to make a comment on our system's expressivity: the strongest  hybrid language employs both the existential binder and satisfaction operator  for state variables (i.e. $@_x$ with $x\in {\rm SVAR}$). Our systems seems to be weaker, but the exact relation will be analyzed elsewhere.

Concerning hybrid modal systems in many-sorted setting, we refer to  \cite{platzerhybrid,rosulics}. The system from \cite{platzerhybrid} is built upon differential dynamic logic, while the one from \cite{rosulics} is equationally developed, does not have nominals and satisfaction operators, the strong completeness being obtained  in the presence of a stronger operator called {\em definedness} (which is the modal global operator). Note that, when the satisfaction operator is defined on state variables, the global modality is definable in the presence of the universal binder. However, we only have the satisfaction operator defined on nominals, so, again, our system seems to be weaker.

There are many problems to be addressed in the future, both from theoretical and practical point of view. We should definitely analyze the standard translation and clarify the issues concerning expressivity; we should  study the Fischer-Ladner closure and analyze completeness w.r.t. standard models from the point of view of dynamic logic; of course we should analyze more  practical examples and even employ automatic techniques.

To conclude, the analysis of  hybrid modal logic in a many-sorted setting leads us to a general system, that is   theoretically   solid and practically flexible enough for our purpose. We were able to specify a programming language, to define its operational semantics and to perform Hoare-style verification, all within the same deductive system. Modal logic proved to be, once more, the right framework and  in the future we hope to take full advantage of its massive development.


\newpage
\begin{thebibliography}{8}

  \bibitem{hand}
Areces, C., ten Cate, B.: \emph{Hybrid Logics}. In: Handbook of Modal Logic, P. Blackburn et al. (Editors) 3, pp. 822--868 (2007).

\bibitem{blsel}
P. Blackburn, J. Seligman, Hybrid Languages, Journal of Logic, Language and Information, (4):251-272 (1995).

 \bibitem{pureax}
Blackburn, P., ten Cate, B.: \emph{Pure Extensions, Proof Rules, and Hybrid Axiomatics}. Studia Logica 84(2), pp. 277--322 (2006).

\bibitem{hyb}
Blackburn, P., Tzakova, M.: \emph{Hybrid Completeness}. Logic Journal of the IGPL 4, pp. 625--650 (1998).

\bibitem{hybtemp}
Blackburn, P., Tzakova, M.: \emph{Hybrid languages and temporal logic}. Logic Journal of the IGPL 7, pp. 27--54 (1999).

\bibitem{mod}
Blackburn, P., Venema, Y., de Rijke, M.: Modal Logic. Cambridge University Press (2002).
  
\bibitem{platzerhybrid}  
Bohrer, B., Platzer, A.: {\emph A Hybrid, Dynamic Logic for Hybrid-Dynamic Information Flow}. In: LICS'18 Proceedings of the 33rd Annual ACM/IEEE Symposium on Logic in Computer Science, pp. 115--124 (2018).
  
\bibitem{popl}
Calcagno, C., Gardner, P., Zarfaty, U.: \emph{Context logic as modal logic: completeness and parametric inexpressivity}. In: POPL'07 Proceedings of the 34th annual ACM SIGPLAN-SIGACT symposium on Principles of programming languages, pp. 123--134 (2007).
  
\bibitem{rosulics}  
Chen, X., Ro\c su, G.: \emph{Matching mu-Logic}.  LICS'19. To appear. Technical report: http://hdl.handle.net/2142/102281 (2019).

\bibitem{floyd}
Floyd, R. W.: \emph {Assigning meanings to programs}. In: Proceedings of the American Mathematical Society Symposia on Applied Mathematics 19, pp. 19--31 (1967).
 
\bibitem{goranko}
Gargov, G., Goranko, V.:\emph{Modal logic with names}. Journal of Philosophical Logic 22, pp. 607--636 (1993).

\bibitem{goranko2}  
Goranko, V., Vakarelov, D.: \emph{Sahlqvist Formulas in Hybrid Polyadic Modal Logics}. Journal of Logic and Computation 11 (2001).

\bibitem{goguen}
Goguen, J., Malcolm, G.: Algebraic Semantics of Imperative Programs.  MIT Press (1996).


\bibitem{dynamic}
Harel, D., Tiuryn, J., Kozen, D.: Dynamic logic. MIT Press Cambridge (2000) 

  
  \bibitem{HHKR89}
Heering, J.,  Hendriks, P.R.H., Klint, P., Rekers, J., \emph{The syntax definition formalism SDF ---reference manual---}. ACM Sigplan Notices 24(11), pp. 43--75 (1989).

\bibitem{hoare}
Hoare, C. A. R.: \emph{An axiomatic basis for computer programming}. Communications of the ACM 12(10), pp. 576--580 (1969). 



\bibitem{noi}
Leu\c stean, I., Moang\u a, N., \c Serb\u anu\c t\u a, T. F.: \textit{A many-sorted polyadic modal logic}. arXiv:1803.09709, submitted (2018).

\bibitem{platzerbook}
Platzer, A.: Logical Foundations of Cyber-Physical Systems. Springer (2018).

\bibitem{plotkin}
Plotkin, G. D.: A Structural Approach to Operational Semantics (1981) Tech. Rep. DAIMI FN-19, Computer Science Department, Aarhus University, Aarhus, Denmark. (Reprinted with corrections in J. Log. Algebr. Program) 60-61, pp. 17--139 (2004).

\bibitem{rosu}
Ro\c su, G.: \emph{Matching logic}. In:  Logical Methods in Computer Science 13(4), pp. 1--61 (2017).


\bibitem{separation}
Reynolds, J. C.:\emph{Separation logic: A logic for shared mutable data structures}. In: Proceedings 17th Annual IEEE Symposium on Logic in Computer Science (2002).

\bibitem{manys6}
Schr\"{o}der, L., Pattinson, D.:
\emph{Modular algorithms for heterogeneous modal logics via multi-sorted coalgebra}. In: Mathematical Structures in Computer Science 21(2) , pp. 235--266 (2011).

  
\bibitem{manys2}
Venema, Y.: \emph{Points, lines and diamonds: a two-sorted modal logic for projective planes}. In: Journal of Logic and Computation, pp. 601--621 (1999).
 
\end{thebibliography}
\end{document}